\documentclass[letterpaper,conference,10pt]{IEEEtran}

\IEEEoverridecommandlockouts
\usepackage{times,bbm}

\usepackage{amsfonts,latexsym}
\usepackage{amsmath}
\usepackage{accents}
\usepackage{amssymb}
\usepackage{graphicx}
\usepackage{amsmath,amssymb}  
\usepackage{comment}
\usepackage{xcolor}
  
\usepackage{xcolor}
\usepackage{lipsum}
\usepackage{cite}

\usepackage{changepage}   % for the adjustwidth environment

\newenvironment{wholeindent}{\begin{adjustwidth}{1.2cm}{}}{\end{adjustwidth}}

\providecommand{\com[2]}{\begin{tt}[#1: #2]\end{tt}}

\newcommand{\del}[1]{{\textcolor{yellow}
{\begin{tt}Suggested deletion: \end{tt}#1}}}
\newcommand{\jmh}[1]{{#1}}
\newcommand{\sam}[1]{\textcolor{black}{#1}}

\newcommand{\samf}[1]{\textcolor{black}{#1}}
\newcommand{\julf}[1]{\textcolor{black}{#1}}
\newcommand{\jul}[1]{\textcolor{black}{#1}}
\newcommand{\julien}[1]{{#1}}
\renewcommand{\P}{\mathbb{P}}
\newcommand{\R}{\mathbb{R}}
\renewcommand{\Re}{\mathbb{R}}

\newcommand{\N}{\mathbb{N}}

\newcommand{\1}{{\mathbf{1}}}
\newcommand{\E}{\mathbb{E}}

\newcommand{\NN}{{\mathcal{N}}}

\newcommand{\un}{{\bf{1}}}

\newcommand{\ie}{{\it i.e. }}

\providecommand{\quotes[1]}{``#1"}

\newcommand{\lp}{e^{-\ratio_p}}

\newcommand{\ssum}{\displaystyle\sum}
\newcommand{\llim}{\displaystyle\lim}

\newcommand{\iSjnotS}{\begin{minipage}{0.6cm}{\begin{center}\tiny{$i\in S$ \\ $j\notin S$}\end{center}}\end{minipage}
}
\newcommand{\Diam}{\Delta}

\newcommand{\ratio}{\rho}

\newcommand{\eps}{\varepsilon}

\renewcommand{\underbar}[1]{{\underaccent{\bar}{#1}}}
\newcommand{\xM}{\bar{x}}
\newcommand{\xm}{\underbar{x}}
\newcommand{\dmin}{\delta_{\min}}
\newcommand{\dmax}{\delta_{\max}}

\newcommand{\sat}{\text{sat}}
\newcommand{\initiate}{\textit{engage}}
\newcommand{\initiates}{\textit{engages}}
\newcommand{\reply}{\textit{reciprocate}}
\newcommand{\replies}{\textit{reciprocates}}
\providecommand{\utij}{\underline t_{ij}}
\providecommand{\otij}{\overline t_{ij}}

\newcommand{\on}{active~}
\newcommand{\ivj}{\{i,j\}}
\newcommand{\bart}{\bar t}

\newcommand{\Ak}{A_k}
\newcommand{\Auk}{A1_k}
\newcommand{\Adk}{A2_k}
\newcommand{\Bz}{B_0}
\newcommand{\Buz}{B1_0}
\newcommand{\Bdz}{B2_0}
\newcommand{\Bk}{B_k}
\newcommand{\Buk}{B1_k}
\newcommand{\Bdk}{B2_k}
\newcommand{\Bkk}{B_{k+1}}
\newcommand{\Bukk}{B1_{k+1}}
\newcommand{\Bdkk}{B2_{k+1}}

\newtheorem{theorem}{Theorem}
\newtheorem{lemma}[theorem]{Lemma}
\newtheorem{proposition}[theorem]{Proposition}

\newtheorem{assumption}{Assumption}
\newtheorem{remark}{Remark}
\newtheorem{observation}{Observation}
\usepackage{algorithm}
\usepackage{algorithmicx}
\usepackage{algpseudocode}

\begin{document}

\pagestyle{empty}

%%%%%%%%%%%%%%%%%%%%%%%%%%%%%%%%%%%%%%%%%%%%%%%%%%%%%%%%%%%%%%%

\title{Continuous-Time Consensus under Non-Instantaneous Reciprocity}

\author{Samuel Martin and Julien M. Hendrickx% <-this % stops a space
\thanks{Julien Hendrickx is with the ICTEAM institute, Universit\'e catholique de Louvain, Louvain-la-Neuve, Belgium.   {\tt\small julien.hendrickx@uclouvain.be}}
\thanks{Samuel Martin is with Universit\'e de Lorraine and CNRS, CRAN, UMR 7039, 2 Avenue de la For\^et de Haye, 54518 Vand\oe uvre-l\`es-Nancy, France (part of the work was carried out when S. M. was with the ICTEAM institute).
        {\tt\small samuel.martin@univ-lorraine.fr } }%
 \thanks{This work is supported by the Belgian Network DYSCO (Dynamical Systems, Control, and Optimization), 
funded by the Interuniversity Attraction Poles Program, initiated by the Belgian Science Policy Office, and by the Concerted
Research Action (ARC) of the French Community of Belgium.
This work is also partly supported by the French Agence Nationale de la Recherche under ANR COMPACS - Computation Aware Control Systems, ANR-13-BS03-004 and by the CNRS via the interdisciplinary PEPS Project MADRES.}
}

%\author[Samuel Martin]{Samuel Martin}
%\address{Large Graphs and Networks Group \\
%Universit\'e Cahtolique de Louvain} \email{samuel.martinsa@gmail.com}

%\author[Julien Hendrickx]{Julien Hendrickx}
%\address{Large Graphs and Networks Group \\
%Universit\'e Cahtolique de Louvain} \email{julien.hendrickx@uclouvain.be}

%%%%%%%%%%%%%%%%%%%%%%%%%%%%%%%%%%%%%%%%%%%%%%%%%%%%%%%%%%%%%%%%

\maketitle

%%%%%%%%%%%%%%%%%%%%%%%%%%%%%%%%%%%%%%%%%%%%%%%%%%%%%%%%%%%%%%%%

\pagestyle{empty} %No headings for the first pages.
%\comjh{possible new abstract}\\

\begin{abstract}
We consider continuous-time consensus systems whose interactions satisfy a form of reciprocity that is not instantaneous, but happens over time. We show that these systems have certain desirable properties: They always converge independently of the specific interactions taking place and there exist simple conditions on the interactions for two agents to converge to the same value. This was until now only known for systems with instantaneous reciprocity.
These results are of particular relevance when analyzing systems where interactions are a priori unknown, being for example endogenously determined or random. We apply our results to an instance of such systems.
\end{abstract}

\section{Introduction}

We consider systems where $n$ agents each have a value $x_i\in \Re$ that evolves according to  
\begin{equation}\label{eq:def_sys_derivative}
\dot x_i = \sum_{j=1}^n a_{ij}(t) (x_j(t) - x_i(t)),
\end{equation}
where the $a_{ij}(t)\geq 0$ are non-negative functions of time. This means that the value of $x_i$ is continuously attracted by the values of the agents $j$ for which $a_{ij}(t)\neq 0$. 
These systems are called consensus systems because the interactions tend to reduce the disagreement between the interacting agents, and because any consensus state where all $x_i$ are equal is an equilibrium of the system. Analogous systems also exist in discrete time 
\cite{jadbabaie2003coordination,tsitsiklis1985problems,Moreau2005}.
% Blondel2005
% [+ref] removed to save some space DeGroot1974,
Consensus systems play a major role in decentralized control \cite{lin2004local},
%\comjh{add some ref?}, 
data fusion  \cite{boyd2006randomized,xiao2005scheme} \jmh{and} distributed optimization %\cite{duchi2012dual,zanella2011newton,nedic2010constrained},
\cite{duchi2012dual,nedic2010constrained}, 
but also when modeling some animal \cite{chazelle2009convergence,Vicsek1995}
 or social phenomena \cite{lorenz2007continuous, 
 castellano2009statistical}. 
 %boundedconf2002
 %ben2003bifurcations
%\jmh{\comjh{should we remove some references from pysics here?}}
% removed :blondel2009krause,lorenz2005stabilization

%, see [...] for survey. 

General convergence results for consensus systems involve connectivity assumptions that are hard to check for state-dependent interactions, and do not allow treating clustering phenomena. As detailed in the state of the art, more recent results guarantee convergence to one or several clusters under various assumptions on the symmetry or reciprocity of the interactions. All these reciprocity properties have however to be satisfied instantaneously and at every time. We extend them to treat systems where reciprocity is not instantaneous but happens on average over time. 
%\sam{This extension is not trivial. In fact, non-instantaneous reciprocity may fail to ensure convergence and lead to oscillatory behavior when the interaction weights are not properly bounded (see Section~\ref{sec:oscillatory-behavior} for an example). To prove our result we identify 
%%the system satisfies the non-instantaneous reciprocity Assumption, there exists 
%a sequence of states 
%%(taken at successive time instants) which can be identified 
%as the trajectory of a converging discrete time consensus system. To do so, we use a simple yet powerful technique which, to our knowledge, has not appeared previously in the literature : to characterize the dynamics of the discrete time system, we consider trivial artificial states being either $0$ or $1$ (see Section~\ref{sec:proof_general}).
%}.

This extension only holds under certain assumptions 
on the way reciprocity occurs.
Indeed, non-instantaneous reciprocity may fail to ensure convergence and lead to oscillatory behaviors when the interaction weights are not properly bounded, or when the time periods across which it occurs grow unbounded (see Section~\ref{sec:oscillatory-behavior} for an example). To prove our result we show that, for an appropriate sequence of times $t_k$, the states $x(t_k)$ can be seen as the trajectory of a certain discrete time consensus system. By analyzing the effect of each matrix of this system on some artificial initial conditions, we obtain bounds on their coefficients, and show that this system satisfies reciprocity conditions guaranteeing convergence.
% for an appropriate sequence $t_k$.} 
%identify 
%the system satisfies the non-instantaneous reciprocity Assumption, there exists 
%a sequence of states 
%(taken at successive time instants) which can be identified 
%as the trajectory of a converging discrete time consensus system.} 
%\del{To do so, we use a simple yet powerful technique which, to our knowledge, has not appeared previously in the literature : to characterize the dynamics of the discrete time system, we consider trivial artificial states being either $0$ or $1$ (see Section~\ref{sec:proof_general}).
%}.

The rest of the paper is organized as follows. The introduction includes a state of the art on consensus systems, a subsection pointing out the interest of non-instantaneous reciprocity and a summary of our contributions. Section~\ref{sec:pb-statement} formally introduces the system that we are considering and presents our main results. Examples illustrating our results and the necessity of an underlying assumption are then presented in Section \ref{sec:examples}. In Section \ref{sec:application}, we demonstrate the use of our results on a specific multi-agent applications. Sections  \ref{sec:proof_sketch} and \ref{sec:proof_pairwise} contain the proofs, and we finish by some conclusions in Section \ref{sec:ccl}.

% and a proof of our main result is sketched in Section \ref{sec:proof_sketch}.
%\del{presents the main result along with examples of systems illustrating when our result applies and why its underlying assumptions are necessary.}

%\comjh{I've changed the structure. We only has one introduction and one section.... Is it ok ? is it a problem that we have again "sketch of proof}

\subsection*{State of the art}

Consensus systems have been the object of many studies during the recent years, focusing particularly on finding conditions under which the system converges, possibly to a consensus state, \samf{and also} on the speed of convergence. Classical results typically guarantee convergence to consensus under some (repeated) connectivity conditions on the interactions,
see for example \cite{Moreau2005,xiao2008asynchronous ,jadbabaie2003coordination}
%\jmh{see the discussion in \cite{Hendrickx2011},} 
% Saber2004, Moreau2004stability
or \cite{Saber2007,Ren2005Survey} for surveys.

\sam{A variation of this repeated connectivity condition was also recently proposed in \cite{ManfrediAngeli2013} for certain classes of state-dependent interactions where the attraction magnitude should be non-decreasing with distance between agents' positions. It involves a graph defined by connecting a node to another when the dynamics of the former is sufficiently and repeatedly influenced by the latter and this being true \samf{for all positions of the two agents}.}
Different recent works have shown that stronger results hold when the interactions satisfy some form of reciprocity.
%something that had already been observed for discrete-time systems
% removed for CDC Lorenz:2003diplomathesis, 
%\cite{Blondel2005, li2004multi,lorenz2005stabilization,Moreau2005}.
%Different recent works 
%\comjh{+ref, including maybe Steve Morse cdc orlando} 
%have however shown that these difficulties could be overcome when the interactions satisfied some form of reciprocity, something that had already been observed for discrete-time systems
% removed for CDC Lorenz:2003diplomathesis, 
%\cite{Blondel2005, li2004multi,lorenz2005stabilization,Moreau2005}.
Hendrickx and Tsitsiklis have for example introduced the \emph{cut-balance} assumption on the interactions \cite{Hendrickx2011}, stating that there exists a $K$ such that for every subset $S$ of agents and time $t$, there holds
\begin{equation}\label{eq:cut-balance}
\sum_{i\in S,j\not\in S}a_{ij}(t) \leq K \sum_{i\in S,j\not\in S}a_{ji}(t).
\end{equation}
This assumption can actually be shown to mean that whenever an agent $i$  influences agent $j$ \jmh{indirectly}, agent $j$ also influences agent $i$ indirectly, with an intensity that is within a constant ratio of that of \jmh{$i$} on \jmh{$j$}. Particular cases of this assumptions include symmetric interactions $a_{ij} = a_{ji}$, bounded-ratio symmetry $a_{ij}\leq K a_{ji}$, or any average-preserving dynamics $\sum_{j} a_{ij} = \sum_{j} a_{ji}$ for every $i$. It was shown in \cite{Hendrickx2011} that systems satisfying the cut-balance assumption \eqref{eq:cut-balance} always converge, though not necessarily to consensus. 
Moreover, two agents' values converge to the same \jmh{limiting} value if they are connected by a path in the graph of \textit{persistent} interactions (also called \textit{unbounded} interactions in the literature), defined by connecting $i$ and $j$ if $\int_{0}^\infty a_{ij}(t)dt$ is infinite.
%Conversely, the values of agents that are not connected by such a path generically converge to different limit.  [should we let the conversely? probably not..]
These results allow analyzing the convergence properties of systems with relatively complex interactions; see the discussion in \cite{Hendrickx2011} for an example in opinion dynamics, or \cite{de2013self} for an application to system involving event-based ternary control of second order agents.

Martin and Girard have later shown~\cite{SamAntoine_Persistent_SICON2013} that in the case of convergence to a global consensus, the cut-balance assumption could be weakened, allowing for the \julien{interaction ratio bound} $K$ to slowly grow with the amount of interactions that have already taken place in the system. They also provide an estimate of the convergence speed in terms of the interactions having taken place.  
%\comjh{[Sam: can you check? + il y avait une autre formalization]} \comsam{autre formalisation ?}

Related convergence results were also proved for systems involving a \julien{continuum} of agents under a strict symmetry assumption in \cite{hendrickx2013symmetric}. %\comjh{peut être virer cette ref la pour la cdc} \comsam{comme tu veux}
\sam{An alternative reciprocity condition called \textit{arc-balance} was considered in \cite{Shi2013}; it requires \emph{all weights} $a_{ij}(t)$ to be within a constant ratio of each other, except those for which $\int_0^\infty a_{ij}(t)dt <\infty$. }
%\sam{
%In the case of \cite{Shi2013}, it is also assumed that the persistent weights all have to be of the same order of magnitude. This implies that all persistent arcs must be non-null simultaneously. Since these persistent weights are assumed to form a graph containing a spanning tree (\ie, there exists a root node from which all other node can be reached), it also prevents to have different time scales in different part of the network. These two limitations do not appear in our result. 
%}

Finally, we note that similar results of convergence under some reciprocity conditions have been obtained for discrete time consensus systems, see for example \cite{touri2014endogenous,touri2012backward,bolouki2013ergodicity, li2004multi,Moreau2005}.
%lorenz2005stabilization
However, none of these results allow for non-instantaneous reciprocity.
%  \jul{(to the exception of a generalization of the results based on arc-balance mentioned in \cite{Shi2013}).

%systemsSimilar 
%(and somewhat stronger \comsam{may be to explain ?}) 
%results were also obtained for discrete time consensus systems, see for example \cite{touri2010,touri2014endogenous,touri2012backward,bolouki2013ergodicity,Blondel2005, li2004multi,lorenz2005stabilization,Moreau2005}
 
% touri2011product removed for space reasons
%\comjh{check those ref+ did not find rabbat}
%[+ ref à proskurnikov]
% [ ...ref Malhame, rabbat, touri]. 
% , which present the additional difficulty that agents can almost entirely \quotes{forget} their previous values in one time-step.
%3 13 14 18 20 21 27
%Product of Random Stochastic Matrices%
%B. Touri and A. Nedi?, Product of Random Stochastic Matrices, to appear in IEEE Transactions on Automatic Control. 

\julien{
\subsection*{Non-instantaneous reciprocity}
}
\label{sec:non-instantaneous-reciprocity}
All the results taking advantage of reciprocity require the reciprocity condition to be satisfied instantaneously \jmh{at (almost) all times}. 
They would thus not apply to systems that are essentially reciprocal, but where the reciprocity may be delayed, or where it happens over time:
\julien{In systems relying on certain wired or wireless network protocols, agents may be unable to simultaneously send and receive information, resulting in loss of instantaneous reciprocity, even if the interactions are meant to be reciprocal.}
Non-instantaneous reciprocity also arises in a priori symmetric systems where the control of the agents is event-triggered or self-triggered. 
%Consider for example a system that is a priori symmetric, but where agents update their control actions in an event-triggered of self-triggered manner, and 
Indeed, suppose that at some time the conditions are such that agents $i$ and $j$ should interact. It is very likely that one agent will update its control action before the other, so that during a certain interval of time the actual interactions will not be symmetric. 

%\comjh{delete next sentence?}(This problem was avoided in \cite{de2013self} because the triggering rule was such that neighboring agents $i$ and $j$ would update their reciprocal interaction exactly at the same time, but this may be a restrictive setting).

Similar problems are present in systems prone to occasional failures, or unreliable communications, where the communication between two agents can temporarily be interrupted in one direction for a limited amount of time.

Issues with non-instantaneous reciprocity may also arise in swarming processes or any multi-agent control problem where sensors have a limited scope.
% due to the limited scope of sensors. 
Suppose indeed that the sensors are not omnidirectional, as it is for example the case for human or animal eyes. It is then generally impossible for an agent to observe all its neighbors at the same time. The same issue arises if the agent can only treat a limited  number of neighbors simultaneously.
A natural solution is then to observe a subset of the neighbors and to periodically modify the subset being observed. This can for example be achieved by continuously rotating the directions in which observations are made. 
%\comjh{as a radar does?}? 
In that case, even if the neighborhood relation is symmetrical, it is again highly likely that an agent $i$ will sometime observe an agent $j$ without that $j$ is observing $i$ at that particular moment, but that $j$ will observe $i$ later.
In all these situations, one could hope to take advantages of the essential reciprocity of the system design even if this reciprocity is not always instantaneously satisfied.

\vspace{0.2cm}

\subsection*{Contributions}

\jmh{We show in our main result (Theorem \ref{th:consensus-under-persistent-connectivity-and-upper-bound-on-cummulative-weights-and-reciprocity})} that the convergence of systems of the form \eqref{eq:def_sys_derivative} is still guaranteed if the system satisfies some form of non-instantaneous reciprocity, or reciprocity on average. More specifically, we assume that the cut-balance condition \eqref{eq:cut-balance} is satisfied \emph{on average} on a sequence of contiguous intervals. These intervals can have arbitrary lengths, but the amount of interaction taking place during each of them should be uniformly bounded. 
Under these assumptions, we show that the system always converges. Moreover, two agent values converge to the same limit if they are connected by a path in the graph of persistent interactions, defined by connecting two agents $i,j$  if $\int_{t=0}^\infty a_{ij}(t)dt$ is infinite.

We also particularize our general result to systems satisfying a form of pairwise reciprocity over bounded time intervals. This particularized result is more conservative,
%applies under stronger conditions
but its condition can often be easier to check.
%often easier to check. %than the general result
We illustrate it on an application.

\section{Problem Statement and Main Results}\label{sec:pb-statement}

%\jmh{\comjh{we do not really state a problem formally. Should we say system description? i'll go with your choice}}

%\julien{We consider the Caratheodory solutions to \eqref{eq:def_sys_derivative}, that are absolutely continuous functions $\Re^+\to \Re^n:t\to x(t)$ satisfying \eqref{eq:def_sys_derivative} at almost all time. Equivalently, Caratheorody solutions are solutions to the integral version of \eqref{eq:def_sys_derivative}}
%}
%\del{In this paper,} 
We study the integral version of the consensus system~(\ref{eq:def_sys_derivative}):
\begin{equation} \label{eq:sys-integral}
x_i(t) = x_i(0) + \int_0^t \ssum_{j=1}^{n}a_{ij}(s)(x_j(s)-x_i(s))ds,
\end{equation}
where for all \julien{$i,j \in \NN=\{1,\dots,n\}$}, the {\it interaction weight} $a_{ij}$ is a non-negative measurable function of time, summable on bounded intervals of $\R^+$. %\julien{(It will be seen that our result also apply to state dependent coefficients $a_{ij}(x,t)$)} 
There exists a unique function of time $x:\R^+ \rightarrow \R^n$ which satisfies for all $t\in \R^+$ the integral equation~(\ref{eq:sys-integral}), \julien{and it is locally absolutely continuous} (see Theorem~54 and Proposition~C.3.8 in~\cite[pages 473-482]{Sontag98}). This function is actually the Caratheodory solution to the differential equation \eqref{eq:def_sys_derivative} and can equivalently be defined as absolutely continuous function
%\comsam{there is only one continuous such function, right ?} \jmh{\comjh{Not exactly sure, there are these annoying Cantor devil's staircases things which can be very surprising. I suggest not insisting on this. Is it a problem to leave it in its present version?}} 
satisfying \eqref{eq:def_sys_derivative} at almost all times. We call it the {\it trajectory} of the system.

%------ see if needed + (does not seem to)

%We say that a trajectory {\it reaches a consensus} when $\lim_{t\rightarrow +\infty} x_i(t)$ exist and are the same, for all $i\in \NN$. The common limit is called the {\it consensus value}.
%\end{definition}
%
%We define the {\it group diameter} as  
%\begin{equation}\label{eq:diameter}
%\Delta_\NN(t) = \xM(t) - \xm(t),
%\end{equation}
%where $\xM(t) = \max_{i\in \NN} x_i(t)$ and $\xm(t) =\min_{j\in \NN} x_j(t)$.
%It can be easily shown that $\xM$ is non-increasing and that 
%$\xm$ is non-decreasing. Then, it is clear that the group diameter is non-increasing and that the trajectory reaches a consensus if and only if 
%$\llim_{t\rightarrow +\infty} \Delta_\NN(t) = 0.$

%\subsection{Main results}
Following the discussion in the Introduction,
%Section~\ref{sec:non-instantaneous-reciprocity}, 
we introduce a new 
%\del{reciprocity assumption which}\comjh{there were too many "reciprocity.." ;-)} 
\jmh{condition} generalizing Condition \ref{eq:cut-balance} by allowing for non-instantaneous reciprocity of interactions; we only require that the reciprocity occurs on the integral weights $\int a_{ij}(s)ds$ over some bounded time intervals.
\vspace{.2cm}
\begin{assumption}[Integral weight reciprocity]\label{as:reciprocity}
%\sam{or delayed or postponed or on average reciprocity?}
There exists a sequence $(t_p)_{p\in \N}$ of increasing times with 
$\lim_{p\rightarrow +\infty} t_p = +\infty$ and some uniform bound $K\ge 1$ such that, for all non-empty proper subsets $S$ of $\NN$, and for all $p \in \N$, \jmh{there holds}
\begin{equation}\label{eq:reciprocity}
\ssum_{i\in S, j\notin S}\int_{t_p}^{t_{p+1}} a_{ij}(t)dt \;\;\le \;\; K\ssum_{i\in S, j\notin S}\int_{t_p}^{t_{p+1}} a_{ji}(t)dt.
\end{equation}
\end{assumption}
\vspace{.2cm}

%\subsection{Relation to existing reciprocity assumption}

%\del{Assumption~\ref{as:reciprocity} generalizes most types of reciprocity found in the consensus literature. In particular, it generalizes the cut-balance assumption~\eqref{eq:cut-balance} developed in~\cite{Hendrickx2011} \jmh{and discussed in the Introduction}.}
%\comjh{I suggest deleting the previous sentence, as it is redundant with the discussion that we have added before assumption 3. I've therefore added the fact that assumption 1 generalized cut-balance just before assumption 1}
%\del{, { specific cases of which include (i) symmetric interactions $a_{ij}=a_{ji}$, (ii) bounded-ratio symmetry $a_{ij}\leq K a_{ji}$ for some $K>0$, or (iii) average-preserving dynamics $\sum_j a_{ij} = \sum_j a_{ji}$ for every $i$.}}\comjh{This was already mentionned in the introduction}.

\julien{We will see in a simple example in Section~\ref{sec:proof-prop-reciprocity-upper-bound-counter-example} that Assumption \ref{as:reciprocity} alone is not sufficient to guarantee the convergence of the system. We need to further assume that the integral of the interactions taking place in each interval $[t_p,t_{p+1}]$ is uniformly bounded.}

%\begin{comment}
%Several classes of consensus systems such as those described in section~\ref{sec:non-instantaneous-reciprocity} present convergent behavior. However, due to asynchronous behaviors, they do not satisfy the cut-balanced assumption nor any instantaneous assumption. Thus, we have to turn to the reciprocity assumption~\ref{as:reciprocity} to prove their convergence. In section~\ref{sec:convergent-toy-system-with-delayed-reciprocity}, we provide an explicit example of such systems.
%\end{comment}
%
%\begin{comment}
%However, when assuming delayed rather than instantaneous reciprocity, convergence is not always guaranteed : the system may instead oscillate indefinitely. To illustrate this fact, we provide a 3-agent example
%in section~\ref{sec:proof-prop-reciprocity-upper-bound-counter-example}.
%Thus, one more assumption has to be force convergence. The integral communication weights must be uniformly upper bounded over intervals where reciprocity occurs.
%\end{comment}
\vspace*{.2cm}
\begin{assumption}[Uniform upper bound on integral weights]\label{as:upper-bound-on-weights}
The sequence $(t_p)$ used in Assumption~\ref{as:reciprocity} is such that 
$$
\int_{t_p}^{t_{p+1}} a_{ij}(t)dt \le M,
$$
holds for all $i,j \in \NN$, $p \in \N$ and some constant $M$.
\end{assumption}
\vspace*{.2cm}

\julien{We  now state our main result, whose proof is presented in Section 
\ref{sec:proof_general}.
}

\vspace*{.2cm}

\begin{theorem}\label{th:consensus-under-persistent-connectivity-and-upper-bound-on-cummulative-weights-and-reciprocity}
Suppose that the interaction weights \jmh{of system \eqref{eq:sys-integral}} satisfy Assumptions~\ref{as:reciprocity} (\jmh{integral} reciprocity) and~\ref{as:upper-bound-on-weights} (upper bound on weight integral). Then, every trajectory $x$ of system \eqref{eq:sys-integral} converges.

Moreover, \julien{let $G=(\NN,E)$ be the graph of persistent weights defined by connecting $(j,i)$ if $\int_0^\infty a_{ij}(t)dt = +\infty$. Then, there is a directed path from $i$ to $j$ in $G$ if and only if there is a directed path from $j$ to $i$, and there holds in that case $\lim_{t\to \infty}x_i(t) = \lim_{t\to \infty}x_j(t)$. }
\end{theorem}
\vspace*{.2cm}

The second part of the theorem implies that there is a local consensus in each strongly connected component\footnote{\sam{Strongly connected components are defined as the classes of equivalence on the node set where node $i$ and $j$ belong to the same class if and only if $i$ and $j$ are connected to each other by at least a path from $i$ to $j$ and a path from $j$ to $i$.}} of the graph $G$ of persistent interactions. \sam{Notice that the second part of the theorem also implies that \samf{each strongly \julf{connected} component is fully disconnected from the others in graph $G$ :} no edge leaves one component to arrive at another. This is due to the reciprocity Assumption~\ref{as:reciprocity}.}
\vspace*{.2cm}

\sam{Assumption~\ref{as:reciprocity} generalizes most (instantaneous) reciprocity conditions available in the literature, including cut-balance, and is thus automatically satisfied by any system satisfying such conditions. It is moreover satisfied by classes of systems subject to some form of reciprocity that is delayed due for example to communication constraints. It applies for instance to systems where agents engage in interaction with a neighbor while the latter may be asleep or already busy %\samf{due to an engagement}\julf{\comjh{not sure. To recheck}} 
\samf{interacting with another agent}. We provide an example of such application in Section~\ref{sec:application}.}

%\comjh{a la reflexion, le paragraphe ci dessous me semble fort négatif par rapport à notre travail sans offrir de réponse vraiment convaincquante}

%\del{Checking whether a given (non-instantaneously reciprocal) system verifies the nontrivial integral condition in Assumption~\ref{as:reciprocity} may appear hard in general. One option is to show that it is implied by the specific reciprocal nature of the system, as done in Section~\ref{sec:application}. Another one is to derive sufficient conditions on the initial configuration which implies that the integral condition remains valid over time (see for instance~\cite{MartinGirardFazeliJadbabaie}).}
%\del{Another (apparent) difficulty comes from the fact that the reciprocity conditions and the time intervals over which it has to be satisfied are global. We therefore now introduce a new local assumption} 

%\comsam{Proposition de modification}
\sam{There are several options to check whether a given (non-instantaneously reciprocal) system verifies the integral condition in Assumption~\ref{as:reciprocity}. One option is to show that it is implied by the specific reciprocal nature of the system, as done in Section~\ref{sec:application}. Another one is to derive sufficient conditions on the initial configuration which implies that the integral condition remains valid over time (see for instance~\cite{MartinGirardFazeliJadbabaie}).}

\sam{The reciprocity conditions and the time intervals over which it has to be satisfied are global. We now introduce a new local assumption} 
%\del{Another reason is that the reciprocity periods are centralized. However, we now show that Assumption~\ref{as:reciprocity} can be satisfied in a local way. To do so, we introduce a new local assumption }
that we will show to \sam{imply}
%be more conservative than 
Assumptions \ref{as:reciprocity} and \ref{as:upper-bound-on-weights} when interactions are bounded. It requires that whenever an agent $j$ influences an agent $i$ at some time $t$, both agents should influence each other with a sufficient strength across a certain time interval around $t$.

\sam{
%\vspace{0.2 cm}
\begin{assumption}[Pairwise reciprocity]\label{as:pairwise_reciprocity}
There exists a constant $\eps>0$ such that for every unordered pair $\{i,j\}$ with $i,j\in\NN $ distinct, there exists a constant $T_{ij}>0$ such that for all $t\ge 0$, if $a_{ij}(t) >0$ or $a_{ji}(t)>0$, then there exists $\underline t_{ij}, \overline t_{ij}$ such that
\\
\noindent a) $\overline t_{ij} - \underline t_{ij}\leq T_{ij}$,
\\
\noindent b) $ t \in [\underline t_{ij}, \overline t_{ij}]$,
\\
\noindent c) $\int_{\underline t_{ij}}^{\overline t_{ij}}a_{ij}(t)dt  \geq \eps$ and $\int_{\underline t_{ij}}^{\overline t_{ij}}a_{ji}(t)dt  \geq \eps$.
\end{assumption}
}

\vspace{0.2cm}

\sam{Assumption~\ref{as:pairwise_reciprocity} provides a way of verifying non-instantaneous reciprocity entirely locally, by considering separately each pair of nodes. For instance, reciprocal weights of type $a_{ij}(t) = 1 + (-1)^{\lfloor t \rfloor}$ and $a_{ji}(t) =  1 + (-1)^{\lfloor \omega t + \gamma\rfloor} $ satisfy the pairwise non-instantaneous reciprocity for any constants $\omega>0, \gamma \in \R$, although one of the weights may be null while the other is not. The following Theorem is proved in Section \ref{sec:proof_pairwise}.}

%\sam{Assumption~\ref{as:pairwise_reciprocity} provides a decentralized way to check for non-instantaneous reciprocity. For instance, reciprocal weights of type $a_{ij}(t) = \sin(t)+1$ and $a_{ji}(t) = \sin(\omega t + \gamma)+1$ satisfy the pairwise non-instantaneous reciprocity for any constants $\omega>0, \gamma \in \R$, although one of the weights may be null while the other is not.} \jul{The following Theorem is proved in Section \ref{sec:proof_pairwise}.}

\vspace{0.2cm}

\begin{theorem}\label{th:pairwise_reciprocity}
Suppose that the interaction weights $a_{ij}(t)$ of system \eqref{eq:sys-integral} satisfy  Assumption \ref{as:pairwise_reciprocity} and are uniformly bounded \sam{above} by some constant $M'$. Then they satisfy Assumptions \ref{as:reciprocity} and \ref{as:upper-bound-on-weights}, and the conclusions of Theorem \ref{th:consensus-under-persistent-connectivity-and-upper-bound-on-cummulative-weights-and-reciprocity} hold. 
\end{theorem}

%\comjh{should we make a theorem out of this proposition?}

%asymptotic consensus takes place among agents belonging to the same connected components of the graph $G$ constructed with edges $(j,i)$ with persistent weights : $\int_0^\infty a_{ij}(t)dt = +\infty$, that is, $\lim_{t\to \infty}x_i 
%if Assumption~\ref{as:persistent-conn} (persistent connectivity) is also verified then the trajectory converges toward a consensus state.
%\end{theorem}

\vspace*{.2cm}

\begin{remark}\label{rem:about_mainresults}
Theorem \ref{th:consensus-under-persistent-connectivity-and-upper-bound-on-cummulative-weights-and-reciprocity} and Theorem \ref{th:pairwise_reciprocity} are stated for systems where the coefficients $a_{ij}(t)$ only depend on time, and the proof of Theorem \ref{th:consensus-under-persistent-connectivity-and-upper-bound-on-cummulative-weights-and-reciprocity} actually uses that fact. However, these results can directly be extended to solutions of systems with state-dependent coefficients $\tilde a_{ij}(t,x)$, with typically $a_{ij}(t,x)$ depending on $x_i$ and $x_j$. Indeed, suppose that $x$ is a solution of
\begin{equation}\label{eq:def_nonlinear}
x_i(t) = x(0) + \int_0^t \tilde a_{ij}(s,x(s))(x_j(s)-x_i(s))ds, 
\end{equation}
then $x$ is also a solution of the linear time-varying systems \eqref{eq:sys-integral} with ad hoc coefficients $a_{ij}(t) = \tilde a_{ij}(t,x(t))$, and Theorem \ref{th:consensus-under-persistent-connectivity-and-upper-bound-on-cummulative-weights-and-reciprocity} applies to that linear time-varying system. 
%\del{Verifying if such nonlinear systems satisfy Condition~\ref{eq:reciprocity} is in general difficult. But we will see on an example in Section \ref{sec:application} how it can be achieved when the structure of the interactions guarantees a sufficient reciprocity. Note also that the existence or uniqueness of a solution to nonlinear systems of the form \eqref{eq:def_nonlinear} is in general a complex issue.
%Similar extensions apply to randomized weights $a_{ij}$.}
\sam{Verifying if such nonlinear systems satisfy Assumption~\ref{as:reciprocity} can be achieved when the structure of the interactions guarantees a sufficient reciprocity. We will see on an example in Section \ref{sec:application} how this can be done. Note also that the existence or uniqueness of a solution to nonlinear systems of the form \eqref{eq:def_nonlinear} is in general a complex issue.
Similar extensions apply to randomized weights $a_{ij}$.}

Finally, one can verify that Theorem \ref{th:consensus-under-persistent-connectivity-and-upper-bound-on-cummulative-weights-and-reciprocity} and Theorem \ref{th:pairwise_reciprocity} can be extended to systems with agent values $x_i$ in $\Re^n$ provided that the weights $a_{ij}$ remain scalar. It suffices indeed in that case to apply the result separately to each component of the states $x_i$.
\end{remark}
\vspace*{.2cm}

\section{\jmh{Examples}}
\label{sec:examples}

\subsection{System with non-instantaneous reciprocity}
\label{sec:convergent-toy-system-with-delayed-reciprocity}

In this subsection, we present two simple 4-agent systems whose convergence can be established by Theorem~\ref{th:consensus-under-persistent-connectivity-and-upper-bound-on-cummulative-weights-and-reciprocity} and by no other result on consensus available in the literature.
%  that does not verify any of the results on consensus available in the literature despite its convergence.  whose convergence can thus only be established thanks to Theorem~\ref{th:consensus-under-persistent-connectivity-and-upper-bound-on-cummulative-weights-and-reciprocity} to show its convergence.

\begin{figure}
\centering
\begin{tabular}{cc}
\includegraphics[scale = .37]{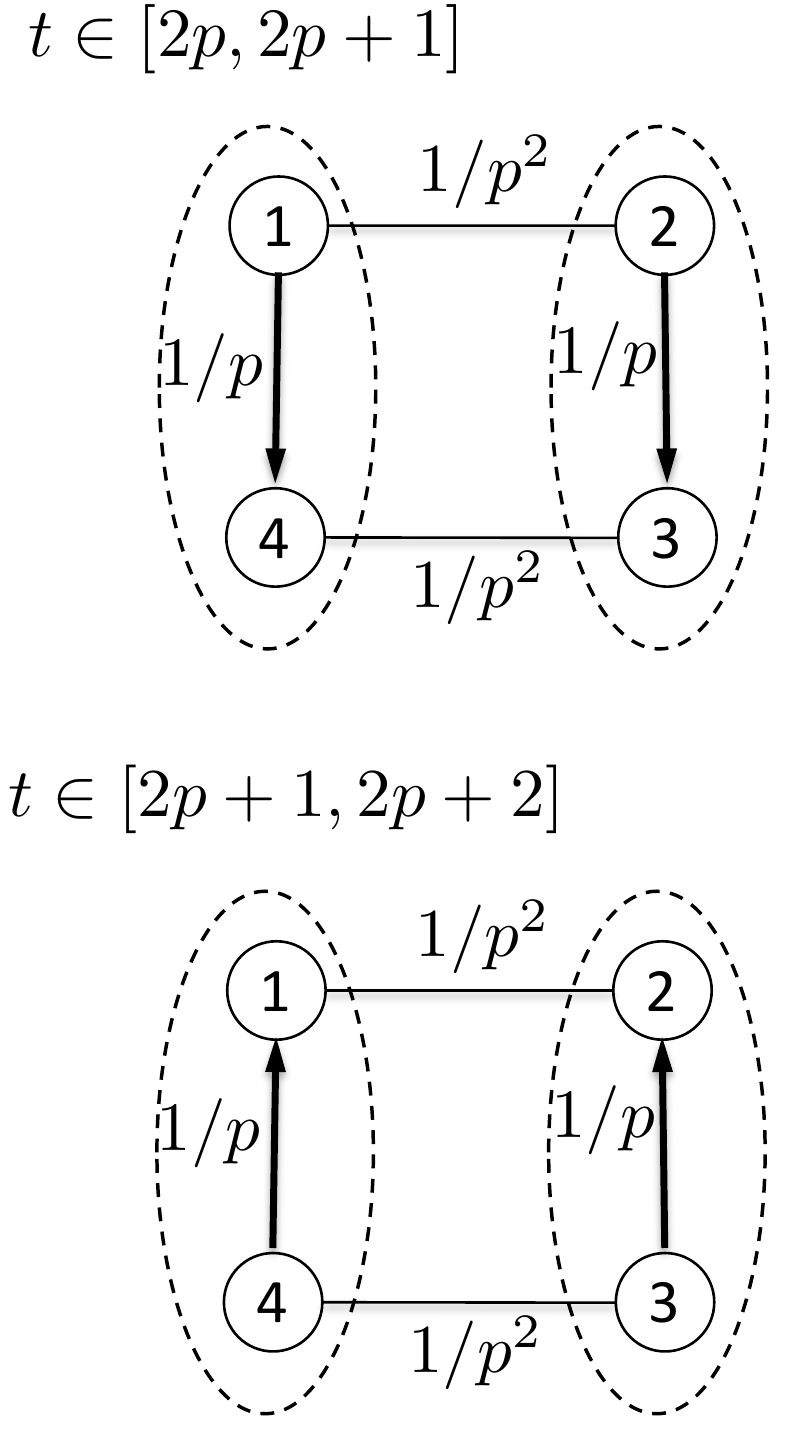}&
\includegraphics[scale = .37]{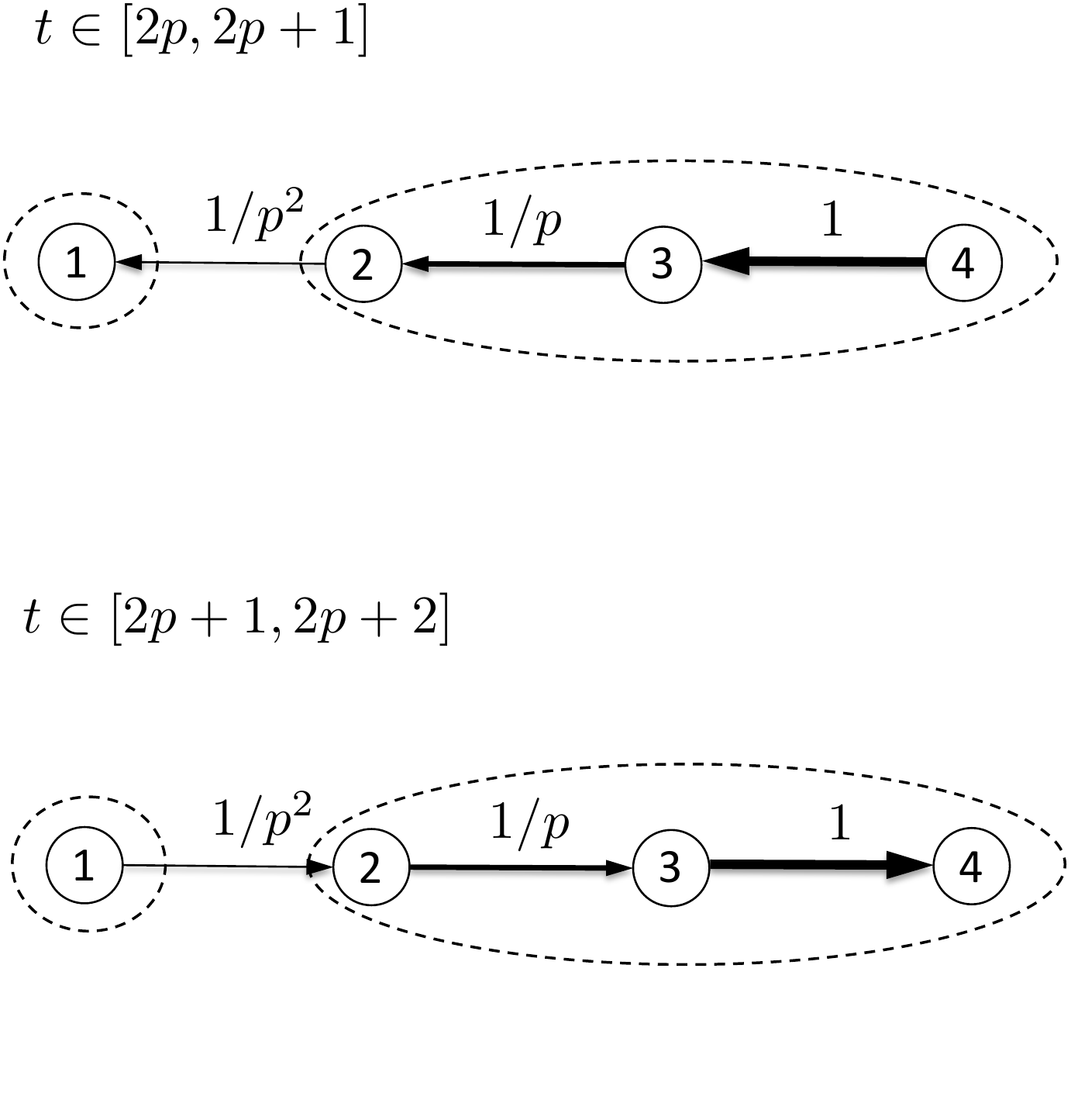}\\
(a) &(b)\end{tabular}\caption{Representations of the interactions taking place in example 1 (a) and in example 2 (b) in Section \ref{sec:convergent-toy-system-with-delayed-reciprocity}, and of the connected components of the graph of persistent interactions, in which local consensuses occur.}
\label{fig:ex}
\end{figure}

%\julf{\comjh{consensuses seems to be the correct word. See e.g. }}
%\begin{verbatim} http://dictionary.reference.com/browse/consensus
%\end{verbatim}

\vspace*{.2cm}
\emph{Example 1:}
\vspace*{.2cm}

%\jmh{\comjh{figure corrected}}

Our first example is depicted in \julf{Fig.} \ref{fig:ex}(a). It contains two weakly interacting subsystems, inside each of which two agents succesively attract each other.
More specifically, the interactions start at time $t=2$ and are defined as follows: For every $p\geq 1$,
\begin{itemize}
\item if $t\in [2p,2p+2]$, $a_{12}=a_{21}=a_{34}=a_{43}= 1/{p^2}$,
\item if $t\in [2p,2p+1]$, $a_{32} = a_{41} = 1/p$,
\item if $t\in [2p+1,2p+2]$, $a_{23} = a_{14} = 1/p$,
\end{itemize}
and all values of $a_{ij}(t)$ that are not explicitly defined are equal to 0.
One can verify that this system satisfies Assumptions \ref{as:reciprocity} and \ref{as:upper-bound-on-weights} with $t_p = 2p$, $K=1$ and $M=2$. We can thus apply Theorem \ref{th:consensus-under-persistent-connectivity-and-upper-bound-on-cummulative-weights-and-reciprocity} to establish its convergence. The graph of persistent interactions can also easily be built and contains the edges $(2,3),(3,2),(1,4)$ and $(4,1)$. There are thus two connected components  $\{2,3\}$ and $\{1,4\}$, and two local consensuses $x_2^*=x_3^*$ and $x_1^* = x_4^*$.
 
On the other hand, notice that the system does not satisfy any instantaneous 
reciprocity condition, so none of available reciprocity-based results applies.
Moreau's result does not apply either due to the weak interactions in $1/p^2$ between the subsystems (the interactions are not lower bounded; see Section 3.3 in~\cite{SamAntoine_Persistent_SICON2013} for a detailed explanation), and because it can only imply convergence to a global consensus while this system produces two local \jmh{consensuses.}
%implies convergence 
%\jmh{\comjh{not sure to fully understand, but ok to keep it as is for the CDC}}
%\comjh{I'm not too sure anymore that the next comment is relevant. feel free to remove it}
%\comsam{i'm happy with it}
Observe \jmh{also} that our result also applies if the interactions are interrupted during arbitrarily long periods. Suppose indeed that the interactions defined above do not take place during the intervals $[2p,2p+1]$ and $[2p+1,2p+2]$ but during the intervals $[p^2,p^2+1]$ and $[p^2 + p ,p^2 + p +1]$. Assumptions \ref{as:reciprocity} and \ref{as:upper-bound-on-weights} still apply with $t_p = p^2$.

\vspace*{.2cm}

\emph{Example 2:}
 \vspace*{.2cm}

The second example involves a chain of four agents, which are attracted by their higher index neighbor for $t\in[2p,2p+1]$ and their lower index neighbor for $t\in [2p+1,2p+2]$, as depicted in \julf{Fig.} \ref{fig:ex}(b). Moreover, the \jmh{ratios} between weights of the different interactions grow unbounded.

Specifically, the interactions start again at $t=2$, and for each $p\geq 1$,
% we have
\begin{itemize}
\item if $t \in [2p,2p+1]$, $a_{12} = 1/p^2$, $a_{23} = 1/p$ and $a_{34} = 1$
\item if $t \in [2p+1,2p+2]$, $a_{21} = 1/p^2$, $a_{32} = 1/p$ and $a_{43} = 1$
\end{itemize}
and all values of $a_{ij}(t)$ that are not explicitly defined are equal to 0. One can verify again that Assumptions \ref{as:reciprocity} and \ref{as:upper-bound-on-weights} hold with $t_p = 2p$, $K=1$ and $M=2$, so that the convergence of the system follows from Theorem \ref{th:consensus-under-persistent-connectivity-and-upper-bound-on-cummulative-weights-and-reciprocity}. The graph of persistent interactions contains the  edges $(2,3),(3,2),(3,4)$ and $(4,3)$, resulting in a local (trivial) consensus of agent $1$, and a consensus between agent 2, 3 and 4.

%\comjh{Can you check if you agree with the second part of the comment?} \comsam{I'm ok with it}
Again, the system satisfies no instantaneous reciprocity condition, so none of available reciprocity-based results applies. Moreover, all the results of which we are aware and that do not rely on reciprocity require the interaction to be bounded from above and from below, and establish convergence to a global consensus (see \cite{Moreau2004stability} for example). Since the \jmh{ratios} between the values of $a_{34},a_{43}$ and $a_{32},a_{23}$ grow unbounded and the system produces again two local consensuses, it would thus be impossible to apply them. This remains the case even if we restrict our attention to the connected component $\{2,3,4\}$ and/or re-scale the values of the coefficients by scaling time.
%\jul{\comjh{est-ce qu'on insiste pas trop sur le global consensus vs local}}

%\begin{comment}The strength of non-zero interactions are constant over time-interval $[4p,4p+1)$. These strength are of two types :
%\begin{itemize}
% \item (P) persistent : $1/p$,
% \item (NP) non-persistent : $1/p^2$.
%\end{itemize}
%Over interval $[4p,4p+1)$, four interactions take place asynchronously each during one unit of time, in this order : $(1,2)$ of type (NP) ;  $(2,3)$ of type (P) ; $(3,4)$ of type (NP) ; $(4,1)$ of type (P). Moreover we assume that interactions $(2,3)$ and $(4,1)$ are completely symmetric while the other two are unidirectional. The symmetry is necessary for the delayed reciprocity Assumptions~\ref{as:reciprocity} to be valid. \\
%stem clearly has a uniform bound on integral weights (Assumption~\ref{as:upper-bound-on-weights}). Thus, it converges in two consensus clusters corresponding to persistent interactions : $\{2,3\}$ and $\{1,4\}$. 
%\end{comment}
%

% It is also impossible to re-scale 

Besides, Theorem 1 would \jmh{again} apply exactly in the same way if the interactions 
%between 1 and 4 were for example evolving 
were interrupted during arbitrary long periods of time
%, as would for example be the case if the interactions 

\subsection{Oscillatory behavior under integral reciprocity - \julien{Necessity of Assumption \ref{as:upper-bound-on-weights}.}}\label{sec:oscillatory-behavior}

%proof of Proposition~\ref{prop:contre-example-ass-persistent-conn-ass-reciprocity}}
\label{sec:proof-prop-reciprocity-upper-bound-counter-example}

The following Proposition formalizes the fact that Assumption \ref{as:reciprocity} alone is not sufficient to guarantee convergence.  
%\ref{th:consensus-under-persistent-connectivity-and-upper-bound-on-cummulative-weights-and-reciprocity} does not hold withou
\vspace{.1cm}
\begin{proposition}
\label{prop:contre-example-ass-persistent-conn-ass-reciprocity}
\jmh{There exist systems of the form \eqref{eq:sys-integral} satisfying 
Assumption~\ref{as:reciprocity} (integral reciprocity) and that admit non-converging trajectories.}
% that do not converge.

%\del{There exist trajectories to system~\eqref{eq:sys-integral} satisfying
%%Assumptions~\ref{as:persistent-conn} (persistent connectivity)
%Assumption~\ref{as:reciprocity} (reciprocity) which do not converge.}
%\comjh{the previous statement was not true 'for every system of the form \eqref{eq:sys-integral}}
\end{proposition}
\vspace{.2cm}
To prove the Proposition, we present a 3-agent system which satisfies Assumption~\ref{as:reciprocity} (reciprocity) but whose trajectory does not converge. 
%Unlike the example of the previous section, this system does not satisfies Assumption~\ref{as:upper-bound-on-weights}. This gives a proof to Proposition~\ref{prop:contre-example-ass-persistent-conn-ass-reciprocity}.
%In this system, 
\julien{The idea is to have \samf{agent 2} oscillating between agents \samf{$1$ and $3$} that successively attract the former while remaining at a certain distance from each other, as} depicted in \julf{Fig.~\ref{fig:dynamic-example}.}
Agent $1$ starts influencing $2$. Since we only impose integral reciprocity, $a_{12}$ and $a_{21}$ do not have to be non-zero simultaneously. Also, because there is no uniform bound on influence, the distance between $2$ and $1$ has become arbitrarily close to $0$ when agent $2$ starts influencing back. So the overall influence of agent $2$ over $1$, this is $\int a_{12}\cdot(x_2 - x_1)dt$ over some time interval, can also be made arbitrarily small.
This leads to an actual influence of $1$ over $2$ but not of $2$ over $1$. The same happens between $3$ and $1$, leading to convergence of $1$ and $3$ to distinct limits and oscillations of $2$. We now present the formal proof.

\begin{proof}
Let $(\ratio_p)_{p\in \N}$ be a non-decreasing sequence such that $\ratio_p\ge 1$, for all $p\in \N$. 
%Let $T>0$ be some constant time duration. Let $\tau = T/4$.
Let us consider a system with $3$ agents where $x_1(0)=0$, $x_2(0)=1/2$ and $x_3(0)=1$ and with the dynamics given by system~(\ref{eq:sys-integral}) with weights
$$
\left\{
\begin{array}{ll}
\text{ if } t\in [4p,4p+1), & a_{21}(t) = \ratio_p, \\
\text{ if } t\in [4p+1,4p+2), & a_{12}(t) = \ratio_p, \\
\text{ if } t\in [4p+2,4p+3), & a_{23}(t) = \ratio_p, \\
\text{ if } t\in [4p+3,4p+4), & a_{32}(t) = \ratio_p,
\end{array}
\right.
$$
where only the non-zero weights have been detailed. \julf{Fig.~\ref{fig:dynamic-example}} illustrates the dynamics of this system.
% of weights always greater than $2$.

\begin{figure}[!htbp]
\begin{center}
%trim=left bottom right top
\includegraphics[scale=0.50,clip = true, trim=0cm 1.8cm 0cm 0cm,keepaspectratio]{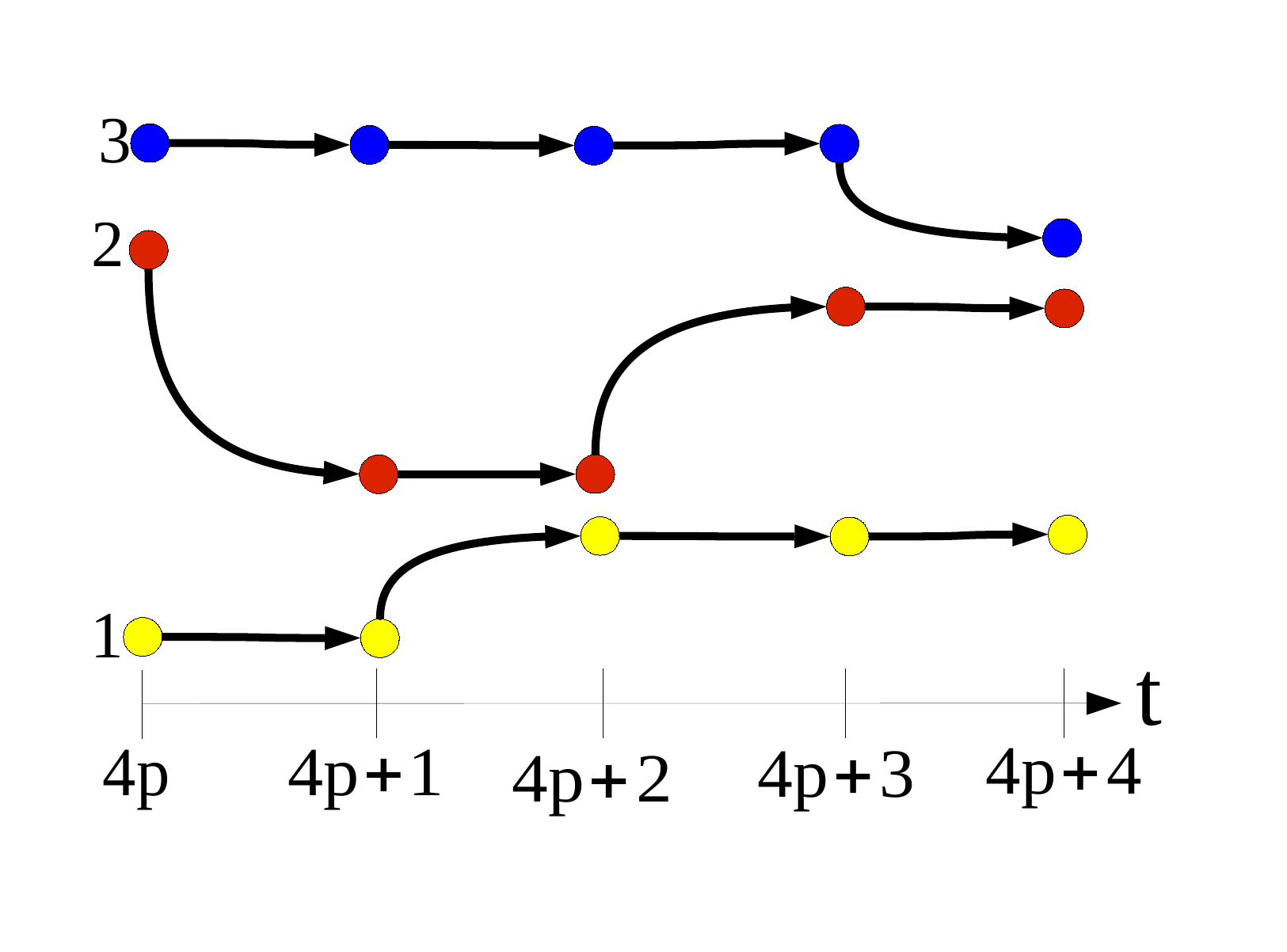}
\caption{\label{fig:dynamic-example}Dynamics of the 3-agent system.}
\end{center}
\end{figure}

%This system satisfies the persistent connectivity Assumption~\ref{as:persistent-conn} (with $(\NN,\tilde{\EE})$ being the undirected line graph). 
Here, Assumption~\ref{as:reciprocity} holds with $K = 1$ \julien{for $t_p = 4p$}.
It is easy to see that $x_1(t)$ is non-decreasing, $x_3(t)$ is non-increasing and $x_1(t)\le x_2(t) \le x_3(t)$ for all $t\ge 0$. Integrating the dynamics of the system, we can show that for all $p\in \N$:
$$
\begin{array}{lll}
x_1(4p+4) &=& x_1(4p+2) \le x_2(4p+2) = x_2(4p+1) \\
&=& (1-\lp) x_1(4p) + \lp x_2(4p) \\
&\le& (1-\lp) x_1(4p) + \lp x_3(0),
\end{array}
$$
%where $\lp = e^{-(\ratio_p)}$, 
and that
$$
\begin{array}{lll}
x_3(4p+4) &\ge& x_2(4p+4) = x_2(4p+3) \\
&=& \lp x_2(4p+2) + (1-\lp) x_3(4p+2) \\
&\ge& \lp x_1(4p) + (1-\lp) x_3(4p) \\
&\ge& \lp x_1(0) + (1-\lp) x_3(4p).
\end{array}
$$
Combining the two previous results and the initial conditions gives us then
\begin{small}
$$
1 + (x_3(4p+4) - x_1(4p+4))
\ge (1-\lp) \left( 1 + (x_3(4p) - x_1(4p)\right).
$$
\end{small}
\noindent 
We observe that the term $1 + (x_3(4p) - x_1(4p))$ remains larger than the product $(1 + (x_3(0) - x_1(0)))\Pi_{p'=0}^p (1- e^{-\rho_{p'}})$.  % la séquence géometrique à facteur variable était un peu ambigue je trouvais
%geometric sequence of scale factor $(1-\lp)$
Taking a sequence $\ratio_p$ growing sufficiently fast (and thus breaking the uniform bound Assumption~\ref{as:upper-bound-on-weights}), one can make this term converge to a value arbitrarily close to its initial value 2.
% leads \jul{thus} to convergence of this term arbitrarily close its initial value $2$. 
Then, $(x_3(4p))$ and $(x_1(4p))$ do not converge to the same value. As a consequence, \jmh{one can verify that} $x_2$ will keep oscillating between $x_1$ and $x_3$. Hence, the system does not converge. 
%\comjh{I've changed the sentence because we do not prove the oscillation of $x_2$ explicitly, although it's quite obvious}
\end{proof}

%\begin{comment}
%It remains to show that the diameter $\Diam_p = x_3(4p) - x_1(4p)$ does not converge towards $0$, in which case it is easy to see that $x_2$ oscillates between the two converging positions of $x_1$ and $x_3$. Thus, the system does not converge. Rewriting the previous equation, we have
%$$
%\Diam_{p+1} + \Diam_0 \ge (1-\lp) (\Diam_p +  \Diam_0).
%$$ 
%We can iterate to obtain for $p\in \N$,
%$$
%\Diam_p \ge (-1+2 \pprod_{q=0}^{p-1} (1-\lqq) ) \Diam_0 \ge (-1+2 \pprod_{q=0}^{\infty} (1-\lqq) ) \Diam_0.
%$$
%Thus, $(\Diam_p)$ does not converge toward $0$ if
%$\prod_{q=0}^{\infty} (1-\lqq) > \frac{1}{2}$ which we can subsenquently rewrite as follows :
%$$
%\begin{array}{lll}
%\ssum_{q=0}^{\infty} \ln (1-\lqq) > \ln\frac{1}{2} &\Leftrightarrow& \ssum_{q=0}^{\infty} -2\frac{\ln (1-\lqq)}{\ln 2} < 2.
%\end{array}
%$$
%Using a well-known result on infinite series, the previous inequality is satisfied by taking the term of the series to be smaller than $1/2^q$ which rewrites to
%$$
%\forall q\in \N, \lqq < 1 - e^{-\frac{\ln 2}{2^{q+1}}}.
%$$
%Thus, choosing sequence $(\rho_q)$ so that
%$$
%\forall q\in \N, \rho_q > -\ln (1 - e^{-\frac{\ln 2}{2^{q+1}}})
%$$
%provides us with an example where the trajectory does not converge despite 
%Assumptions~\ref{as:persistent-conn} (persistent connectivity) and~\ref{as:reciprocity} (reciprocity) being satisfied.
%\end{comment}

\section{Application to mobile robots with intermittent ultrasonic communication}
\label{sec:application}

In this section we apply our results to a realistic system of mobile robots evolving in the plane $\Re^2$ and communicating using ultrasonic sensors. These sensors make for an affordable and thus widespread contactless mean of measuring distances~\cite{Carullo2001}, but are subject to certain limitation as detailed below.
%\comsam{Here add references on ultrasonic communication} \comjh{Tu en connais?}. 
The objective of the group of robots is to achieve practical rendezvous, \ie all robots should eventually lie in a ball of a certain maximal radius (see e.g.\cite{ceragioli2011discontinuities}). The robots have several functional constraints. The ultrasonic sensors in use are not accurate when measuring distances smaller than a radius $d_0>0$, thus we assume that the robots cannot make use of such measurements and are blind at short range. Also, the robots' engines are limited and the velocity of each robot cannot exceed a maximum of $\mu>0$ in norm. Most importantly, in order to save energy, the robots activate their sensors intermittently, and in an asynchronous way:
Robot $i$ wakes up at every time $t_k^i$, and monitors its environment over the time-interval 
$[t_k^i,t_k^i+\delta_{\min}]$, 
for some $\delta_{\min}>0$.
(For simplicity, we take the same $\delta_{\min}$ for every robot, but this is not crucial for our result). In addition, we assume that the sequence $(t_k^i)$ satisfies $t_{k+1}^i - t_k^i \in [\dmin, \dmax]$ for every $k\in \N$, for some $\dmax > \dmin$, and $t_0^i  \leq \dmax$.

%\comjh{changed the order to better separate the system description from the algorithm}

%where the sequence $(t_k^i)$ satisfies $t_{k+1}^i - t_k^i \in [\dmin, \dmax]$

%For a robot $i$, the sequence $(t_k^i)$ describes the time instants where the robot starts monitoring its environment using its ultrasonic sensor. The monitoring occurs over time intervals $[t_k^i, t_k^i + \dmin]$ with $\dmin>0$ some given constant, and we assume the sequence satisfies that for all $k\in \N$,
%$$
%t_{k+1}^i - t_k^i \in [\dmin, \dmax],
%$$
%where $\dmax>\dmin$ and $t_0^i \le \dmax$.

We will provide a simple control law for the robots ensuring some form of non-instantaneous reciprocity. Our result in Section~\ref{sec:pb-statement} will then allow us to establish (i) the convergence of all robot positions, and (ii) asymptotic practical consensus, that is, 
all robots eventually lie at a distance from each other smaller than a certain threshold. This threshold is proportional to $d_0$,
%consensus on the positions up to a difference that is bounded proportionally to $d_0$, 
the distance below which robots cannot sense each other. Since it converges, the system will not suffer from infinite oscillatory behaviors as in the example presented in Section~\ref{sec:proof-prop-reciprocity-upper-bound-counter-example}. To the best of our knowledge, such results cannot be obtained with any other convergence result available in the literature. \sam{One reason for this is that most results on consensus in the literature apply to systems which converge to a single consensus. This is clearly not the case for the system considered here since agents stop interacting at short distance.}
%\comsam{Add a comment explaining why Moreau or other results don't work here, see comment in the reply to the reviewers.}

%\del{Despite these constraints, our result in Section~\ref{sec:pb-statement} allow us to provide a simple algorithm guaranteeing (i) the convergence of all robots, and (ii) practical consensus, that is, that the distance between any two robots is eventually upper bounded proportionally to $d_0$, the distance below which robots cannot sense each other.} 
%\comjh{on n'a pas dit dans quel espace était les positions des robot. Le plus simple est de dire dans Re mais que les résultats s'appliquent dans Re n, mais alors pourquoi utiliser des normes? Sinon on peut dire dans R.2}

%the robots will still achieve rendezvous. In particular, the present settings does not lead to oscillatory behavior as seen in the example presented in Section~\ref{sec:proof-prop-reciprocity-upper-bound-counter-example}. The main results presented in Section~\ref{sec:pb-statement} will enable us to establish this fact.

%For a robot $i$, the sequence $(t_k^i)$ describes the time instants where the robot starts monitoring its environment using its ultrasonic sensor. The monitoring occurs over time intervals $[t_k^i, t_k^i + \dmin]$ with $\dmin>0$ some given constant, and we assume the sequence satisfies that for all $k\in \N$,
%$$
%t_{k+1}^i - t_k^i \in [\dmin, \dmax],
%$$
%where $\dmax>\dmin$ and $t_0^i \le \dmax$. 

Our control law can be expressed as the following saturated consensus equation: %To achieve their objective, the robots follow a saturated consensus system :
\begin{equation}\label{sys:sat-consensus}
 \dot{x}_i(t) = \sat \ssum_{j\in \NN} b_{ij}(t) (x_j(t)-x_i(t)),
\end{equation}
where the $b_{ij}(t)$ will be specified later, and the function $\sat : \Re^n\to \Re^n$ is defined by
\begin{equation*}%\label{eq:sat-function}
\sat(x) = 
\left\{
\begin{array}{ll}
 \mu \cdot \frac{x}{\|x\|} & \text{if } \|x\| \ge \mu \\
 x & \text{otherwise.}
\end{array}
\right.
\end{equation*}
The saturation guarantees that the \samf{magnitude of} the velocity of each robot remains below its limit.
We now explicit how the interaction weights $b_{ij}$ are set.  
%By default, $b_{ij}(t)=0$. Now, we consider two cases. 
The idea is represented in \julf{Fig.} \ref{fig:application}: \sam{For $t \in [t_k^i,t_k^i+\delta_{\min}]$, agent $i$ monitors its environment. At this time, agent $i$ sets $b_{ij}(t)$ to 1 whenever either one of the two following situations occurs : 1) its distance to $j$ is larger than some appropriate radius $d_1>d_0$ (\initiate), or 2) its distance to $j$ is larger than $d_0$ and $j$ has recently been influenced by $i$ ($b_{ji} = 1$) because $j$ was at a distance larger than $d_1$ from $i$ at that time (\reply)}. The latter part of the algorithm is designed to ensure reciprocity, and the presence of $d_1$ is needed to ensure that $i$ and $j$ remain sufficiently distant for measurement to be made when $i$ or $j$ need to reciprocate.

\begin{figure}
\centering
%trim=left bottom right top
\includegraphics[scale = .32,trim=0cm 2cm 0cm 0.5cm]{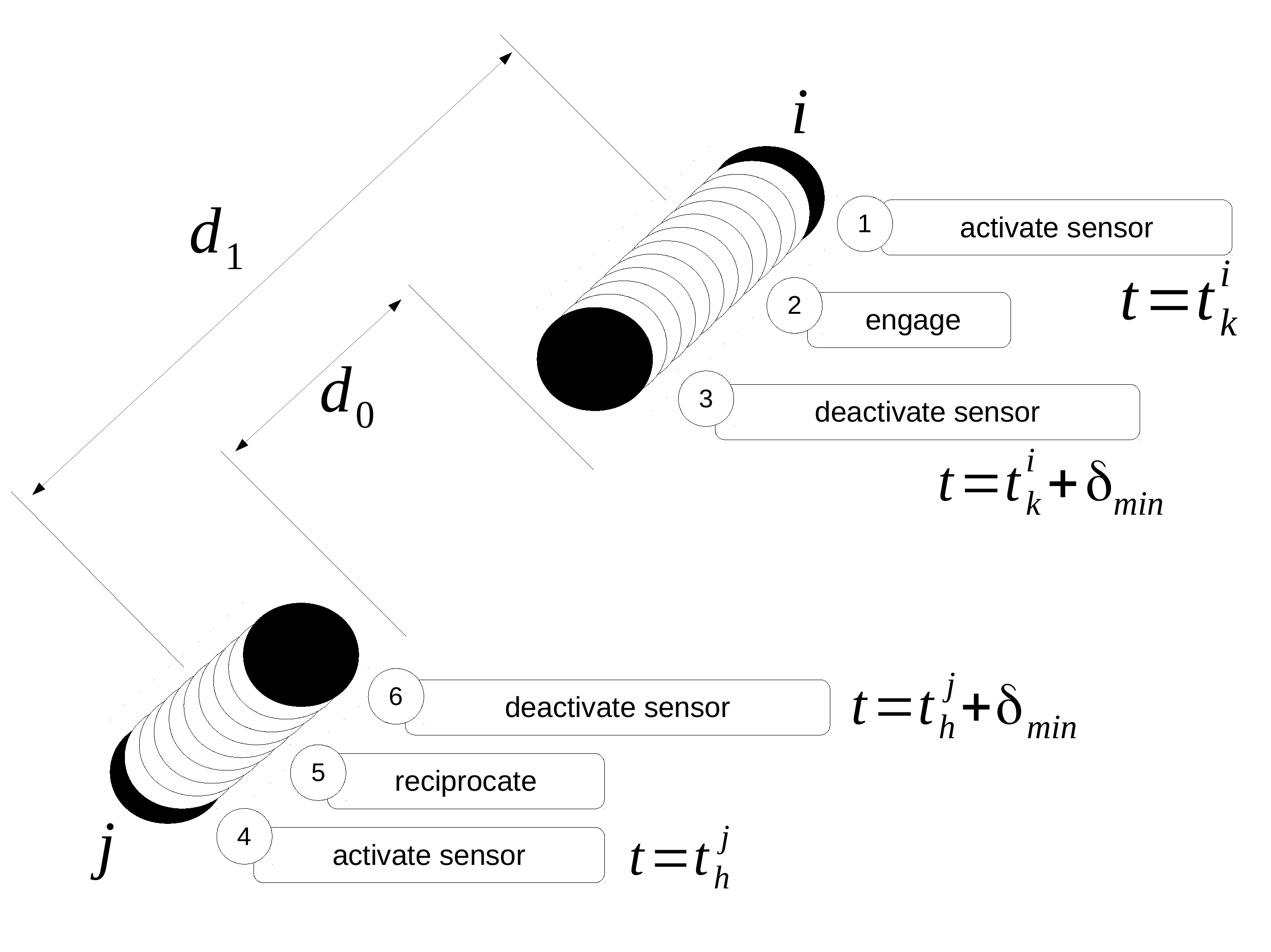}
\caption{Representations of the interactions taking place in the group of mobile robots with intermittent ultrasonic communication presented in Section \ref{sec:application}. Events $1$, $2$ and $3$ occur successively and so do events $4$, $5$ and $6$. Event $4$ occurs after event $1$ and the following condition holds : $t_h^j \in [t_k^i, t_k^i+\dmax]$. When $2$ occurs, $a_{ij}(t) >0$ and when $4$ occurs $a_{ji}(t)>0$. Proposition~\ref{prop:application} provides conditions which guarantee that event $4$ always takes place when event $2$ has occurred, this ensures interaction reciprocity.}
\label{fig:application}
\end{figure}
%\comjh{moved figure in this section, it was hard to understand in the introduction}

Formally, we set $b_{ij}(t)=0$ by default, and set it to 1 in two cases:\\
%
%\begin{equation}\label{eq:aij-is-one}
% b_{ij}(t) = 1
%\end{equation}
\sam{
$i$ \initiates
\begin{equation}\label{eq:cond-initiate}
\begin{array}{l}
\exists k \in \N, \left( t \in [t_k^i, t_k^i + \dmin] \text{ and } \|x_i(t_k^i)-x_j(t_k^i)\| \ge d_1 \right),
\end{array}
\end{equation}
$i$ \replies
\begin{equation}\label{eq:cond-reply}
\begin{array}{l}
\exists h \in \N,\\
\left\{
\begin{array}{l}
t \in [t_h^j, t_h^j + \dmin] \text{ and } \|x_i(t)-x_j(t)\| \ge d_0  \text{ and}
\vspace{0.2cm}
\\
\exists k \in \N, \\
 t^i_k \in [t_h^j - \dmax, t_h^j] \text{ and } \|x_i(t_k^i)-x_j(t_k^i)\|\ge d_1 .
%\|x_j(t_h^j)-x_i(t_h^j)\|\ge \d_0 + \dmin \cdot \mu 
\end{array}
\right.
\end{array}
\end{equation}
}

\begin{remark}
Condition~\eqref{eq:cond-initiate} can be easily implemented. To implement Condition~\eqref{eq:cond-reply}, $i$ has to keep in memory the last activation time $t_h^j$ at which the distance between $i$ and $j$ was higher than $d_1$. This could for example be achieved by having $j$ sending a message to $i$ at $t_h^j$.
\end{remark}

\vspace{.1cm}
Under these communication rules, we have the desired result :

%\comjh{we seem to imply that the  position x is in Ren , but our results are only stated for x in R. This should be clarified somehow. To check..}
\vspace{.2cm}
\begin{proposition}\label{prop:application}
Consider system~\eqref{sys:sat-consensus} where interaction occurs according to Conditions~\eqref{eq:cond-initiate} and~\eqref{eq:cond-reply}. Also assume there holds
\begin{equation}\label{eq:relation-dmax-mu-d1-d0}
4 \dmax \cdot \mu \le d_1 - d_0. 
\end{equation}
Then, the group of robots asymptotically achieves practical rendezvous: $x_i^*=\lim_{t\to\infty}x_i(t)$ exists for every $i\in \NN$, and  
$$
\llim_{t \rightarrow \infty} \Delta(t) \le d_1,
$$
where $\Delta(t) = \max_{i,j\in \NN}||x_i(t)-x_j(t)||$.
\end{proposition}
\vspace{.2cm}

\begin{proof}
Observe first that system~\eqref{sys:sat-consensus} can be rewritten under the form of system~\eqref{eq:sys-integral} with
\begin{equation}\label{eq:aij-bij}
a_{ij}(t) = \frac{\mu \cdot b_{ij}(t)}{\|\ssum_{k \in \NN} b_{ik}(t) (x_k(t) - x_i(t)) \|}
\end{equation}
if
$
\|\sum_{k \in \NN} b_{ik}(t) (x_k(t) - x_i(t))\| \ge \mu
$
and $a_{ij}(t) = b_{ij}(t)$ otherwise. Since $b_{ik}(t)=0$ whenever $\|x_k(t) - x_i(t)\| < d_0$, $a_{ij}$ is upper bounded and thus is a non-negative measurable function, summable on bounded intervals of $\R^+$.

Moreover, since $\Delta(t)= \max_{i,j\in \NN}||x_i(t)-x_j(t)||$ is clearly nonincreasing, it follows from the definition of $a_{ij}(t)$ that 
\begin{equation}\label{eq:bound_aij/bij}
a_{ij}(t) \geq b_{ij}(t)\min\left(\frac{\mu}{n \Delta(0)},1\right),
\end{equation}
where $\Delta(0)$ is the initial group diameter.

In order to apply Theorem \ref{th:pairwise_reciprocity}, we now show that the system under intermittent ultrasonic communication described above satisfies Assumption~\ref{as:pairwise_reciprocity} with 
%$$\eps = \frac{\dmin \mu}{n \Delta(0)},
%$$
%????
$$
\eps = \min\left(\frac{\dmin \mu}{n \Delta(0)},\dmin\right) \text{ and } T =2\dmax.	
$$

%To end the proof, we apply Theorem \ref{th:pairwise_reciprocity} and the system converges. Finally, we establish that convergence can only occurs in a rendezvous state.

Let $t\ge 0$ such that $a_{ij}(t)>0$. Then, $b_{ij}(t)>0$ and at least one among Conditions~\eqref{eq:cond-initiate} and~\eqref{eq:cond-reply} is satisfied. Suppose first that Condition~\eqref{eq:cond-initiate} is satisfied and denote by $k$ the integer such that $t \in [t_k^i, t_k^i + \dmin]$. Clearly, Condition~\eqref{eq:cond-initiate} also holds for every $s\in [t_k^i, t_k^i + \dmin]$. 

We set $\underline t_{ij} = t_k^i$ and $\overline t_{ij}= t_k^i+2 \dmax \geq t_k^i + \dmin$. Clearly, there holds 
$t\in [\utij,\otij],$ and $\otij-\utij \leq 2 \dmax = T$, so that Conditions (a) and (b) of Assumption \ref{as:pairwise_reciprocity} hold.
Moreover, the non-negativity of $a_{ij}$ implies that
\begin{eqnarray*}
\int_{\utij}^{\otij}a_{ij}(s)ds &\geq & \int_{t_k^i}^{t_k^i+\dmin}a_{ij}(s)ds 
\\ 
&\geq & \min\left(\frac{\mu}{n \Delta(0)},1\right) \int_{t_k^i}^{t_k^i+\dmin}b_{ij}(s)ds\\&=& \min\left(\frac{\dmin \mu}{n \Delta(0)},\dmin\right) = \eps,
%&\geq & \sam{\min\left(\frac{\dmin \mu}{n \Delta(0)},\dmin\right)} \\ &=& \eps,
\end{eqnarray*}
where we have used \eqref{eq:bound_aij/bij} and the fact that $b_{ij}(s)=1$ for all $s\in [t_k^i, t_k^i + \dmin]$ since we have seen that Condition~\eqref{eq:cond-initiate} holds for those values. There remains to prove that $\int_{\utij}^{\otij}a_{ji}(s)ds\geq \eps$.

Since  $t_{h+1}^j - t_h^j \le \dmax$ for all $h\in \N$ and $t_0^i \le \dmax$, there exists $h\in \N$ such that $t^j_h \in [t_k^i, t_k^i + \dmax]$, and thus $[t_h^j, t_h^j +\dmax] \subseteq [t_k^i,t_k^i+2\dmax] = [\utij,\otij]$.
We show that the \reply~Condition \eqref{eq:cond-reply} is satisfied for every  $s\in[t_h^j, t_h^j +\dmax]$. The second part of the condition directly follows from $t_h^j\in [t_k^i, t_k^i + \dmax]$. For the first one, observe that $||\dot x_i||\leq \mu$ (and the same holds for $j$), and that $||x_i(t_k) - x_j(t_k)||\geq d_1$ by assumption. Therefore, for any time $s\in [t_h^j, t_h^j +\dmax]\subseteq [t_k^i,t_k^i+2\dmax] $, 
we have
%the distance 
% observe that $||\dot x_i||\leq \mu$, and that $||x_i(t_k) - x_j(t_k)||\geq d_1$ by assumption. Therefore, since $[t_h^j, t_h^j +\dmax] \subseteq [t_k^i,t_k^i+2\dmax] $, 
%we have that 
 %is at least 
\begin{eqnarray*}
 ||x_i(s) - x_j(s)|| &\ge& ||x_i(t_k) - x_j(t_k)||- 4\mu\dmax \\
 &\geq& d_1 - (d_1-d_0)=d_0
\end{eqnarray*}
for every $s\in [t_h^j, t_h^j +\dmax]$, 
where we have used \eqref{eq:relation-dmax-mu-d1-d0}. As a consequence, the first part of Condition \eqref{eq:cond-reply} also holds,
implying that 
$b_{ij}(s)=1$ for every $s\in [t_h^j, t_h^j +\dmax]$. 
We get again
$$\int_{\utij}^{\otij}a_{ji}(s)ds\geq \min\left(\frac{\mu}{n \Delta(0)},1\right)\int_{t_h^j}^{t_h^j+\dmin}b_{ij}(s)ds =\eps,$$
which 
%\samf{terminates the proof of Assumption \ref{as:pairwise_reciprocity}}
%in that case. 
\julf{establishes that Assumption \ref{as:pairwise_reciprocity} holds in that case.}\\

\jmh{Suppose now that $a_{ij}(t)>0$ because Condition (\ref{eq:cond-reply}) is satisfied at $t$ for $i,j$. Then one can easily verify that Condition (\ref{eq:cond-initiate}) was satisfied for $j,i$ for all $s\in [t_j^k,t_j^k+\dmin]$ for some $t_j^k \in [t-\dmax,t]$, and an argument symmetric to that we have developed above shows that Assumption \ref{as:pairwise_reciprocity} also holds.}\\

Since the weights $a_{ij}(t)$ are upper-bounded, applying Theorem \ref{th:pairwise_reciprocity} (or more precisely its direct extension to $\Re^2$, see Remark \ref{rem:about_mainresults})  shows that (i) the system converges: $x_i^*=\lim_{t\to\infty} x_i(t)$ exists for every $i$, and (ii) $x_i^*\neq x_j^*$ only if $\int_0^\infty a_{ij}(t) dt < \infty$.

To conclude the proof, suppose, to obtain a contradiction, that $\lim_{t\to\infty }\Delta(t) > d_1$, and thus that $||x_i^*-x_j^*|| > d_1$ for some $i,j$. 
The continuity of $x$ implies that $||x_i(t) - x_j(t)|| > d_1$ for all $t>s$ for some $s$, and in particular for all $t_k^i>s$. It follows then from the \initiate~rule (\ref{eq:cond-initiate}) that $b_{ij}(t)$ would be set to 1 on infinitely many time intervals of length at least $\dmin$. Besides, it follows from \eqref{eq:bound_aij/bij} that $a_{ij}$ and $b_{ij}$ remain within a bounded ratio, so that we would have $\int_0^\infty a_{ij}(t) dt = \infty$. However, we have seen that  $x_i^*\neq x_j^*$ only if $\int_0^\infty a_{ij}(t) dt < \infty$, so there should hold $x_i^*=x_j^*$, in contradiction with our hypothesis. We have thus $\lim_{t\to\infty }\Delta(t) \leq d_1$.
\end{proof}

\vspace{.1cm}

Note that it is actually possible to have the robots converging to final positions within distances smaller than the $d_1$ from Proposition \ref{prop:application} from each other. This can be achieved by decreasing their maximal speed $\mu$ and the distance $d_1$ when approaching convergence. Such more evolved control laws are however out of the scope of this section, where our goal was to demonstrate the use of our results from Section \ref{sec:pb-statement}.

%\comjh{j'ai rajouté le commentaire ci-dessus au cas ou on trouve que d1 serait pas terrible comme distance de convergence}

\section{Proofs}
% of Theorem~\ref{th:consensus-under-persistent-connectivity-and-upper-bound-on-cummulative-weights-and-reciprocity}}}
\label{sec:proof_sketch}

%%%%%%%%%%%%%%%%%%%%%%%%%%%%%%%%%%%%%%%%%%%%%%%%%%%%%%%
%%%%%%%%%%%%%%%%%%%%%%%%%%%%%%%%%%%%%%%%%%%%%%%%%%%%%%%
\sam{\subsection{Proof of Theorem \ref{th:consensus-under-persistent-connectivity-and-upper-bound-on-cummulative-weights-and-reciprocity}}
\label{sec:proof_general}}
%%%%%%%%%%%%%%%%%%%%%%%%%%%%%%%%%%%%%%%%%%%%%%%%%%%%%%%

Before we prove Theorem~\ref{th:consensus-under-persistent-connectivity-and-upper-bound-on-cummulative-weights-and-reciprocity}, we provide several intermediate results. \julien{Our proof uses the following result on cut-balance discrete-time consensus systems.} \sam{This result is a special case of Theorem~1 in~\cite{touri2014endogenous} restricted to deterministic systems.}
\vspace{.2cm}
%\addtocounter{theorem}{1}

\begin{theorem}\label{thm:DT_result}
Let $y:\N\to \Re^n$ be a solution to
\begin{equation}\label{eq:iter_DT}
y_i(p+1) = \sum_{j=1}^n b_{ij}(p) y_j(p),
\end{equation} 
where $b_{ij}(p)\geq 0$ and $\sum_{j=1}^n b_{ij}(p) =1$.
Suppose that the following assumptions hold:\\
%\vspace{.2cm}

a) \emph{Lower bound on diagonal coefficients:} There exists a $\beta >0$ such that $b_{ii}(p) \geq \beta$ for all $i,p$. 
\vspace{.2cm}

b) \emph{Cut balance:} There exists a $K'>0$ such that for every $p$  and non-empty proper subset $S$ of $\NN$, there holds
\begin{equation}\label{eq:cut_balance_assumption_DT}
\sum_{i\in S, j\not \in S} b_{ij}(p) \leq K' \sum_{i\in S, j\not \in S}b_{ji}(p).
\end{equation}
\vspace{.2cm}

Then, $y^*_i = \lim_{p\to \infty} y_i(p)$ exists for every $i$. Moreover, let $G' = (\NN,E')$ be a directed graph where $(j,i)\in E'$ if $\sum_{p=0}^\infty b_{ij}(p) = +\infty$. %Then every weakly connected component of $G$ is strongly connected [remind what it is??], and there holds $y_i^*= y_j^*$ if $i$ and $j$ belong to the same connected component, that is, there is a directd path from $i$ to $j$ (or equivalently from $j$ to $i$).
There is a path from $i$ to $j$ in $G'$ if and only if there is a path from $j$ to $i$, and in that case there holds  $y_i^* = y_j^*$.
\end{theorem}
%\comjh{Problem here, I think that this result is actually more or less a particular case of the result from \cite{touri2014endogenous}.}
\vspace{.2cm}
Unlike certain results pre-dating those in \cite{touri2014endogenous} (e.g. Theorem 2 in \cite{hendrickx2011new}), Theorem \ref{thm:DT_result} does not require the existence of a uniform lower bound on the positive coefficients $b_{ij}$, that is, the existence of a $\beta'$ such that $b_{ij}(p)>0\Rightarrow b_{ij}(p)\geq \beta'$. This seemingly minor difference is actually essential for our purpose, as there is in general no such uniform lower bound in the context of our proof.

To apply Theorem~\ref{thm:DT_result}, we focus on the values taken by the states at times $t_p$.
Remember that the sequence of times $t_p$ defines the intervals over which the integral reciprocity is satisfied.
\vspace{.2cm}
\begin{lemma}\label{l:discrete-time-system}
%Consider a trajectory to system~(\ref{eq:sys-integral}) defined on $\R$ with given weights $a_{ij}(t)$ and sampling times $t_p$.
The sequence of states $(x(t_p))$ can be written as the trajectory of the discrete-time consensus system obtained by sampling~\eqref{eq:sys-integral}
\begin{equation}\label{eq:discrete_process_phi}
x_i(t_{p+1}) = \ssum_{j \in \NN} \phi_{ij}(p) \cdot x_j(t_p),
\end{equation}
where the weights $\phi_{ij}(p)$ are non-negative and satisfy $\sum_{j \in \NN} \phi_{ij}(p)= 1$. This sampled system always exists and is unique for given weights $a_{ij}(t)$ and sampling times $t_p$. The weights  $\phi_{ij}(p)$ are independent of states $x(t)$. 

In particular, if $x_j(t_p) = 1$ for $j\in S$ and $x_k(t_p) = 0$ for $k\notin S$, for some $S \subseteq \NN$, there holds
\begin{equation}\label{eq:phi_artificial_states}
\ssum_{j\in S} \phi_{ij}(p) =  x_i(t_{p+1}).
\end{equation}
%we can always compute  weights $\phi_{ij}(p)$ by assuming artificial states $x_j(t_p) = 1$ for $j\in S$ and $x_k(t_p) = 0$ for $k\notin S$ and $S \subseteq \NN$, then
%\begin{equation}\label{eq:phi_artificial_states}
%\ssum_{j\in S} \phi_{ij}(p) =  x_i(t_{p+1}).
%\end{equation}
\end{lemma}
\vspace{.2cm}
\begin{remark}\label{rem:artificial}
The equality \eqref{eq:phi_artificial_states} provides a way of computing or bounding certain sums of the weights $\phi_{ij}(p)$ by considering the evolution of the systems starting from \quotes{artificial} states, where $x_j(t_p) = 1$ for some agents and $x_k(t_p) = 0$ for the others. 

Note that these artificial states are only a formal tool to compute weights $\phi_{ij}(p)$, and their use does not result in any loss of generality.
\end{remark}
\vspace{.2cm}

\begin{proof}
%[of Lemma~\ref{l:discrete-time-system}]
Denote by $\Phi(t,T)$ the fundamental matrix of the linear dynamics~(\ref{eq:sys-integral}) which is uniquely defined~\cite{Filippov1988} by
$$
x(T) = \Phi(t,T) x(t).
$$
%$D(s)=diag(d_i)$ the degree diagonal matrix of $A(s)$ with $d_i(s) = \sum_{j\in \NN} a_{ij}$. 
We define $\phi_{ij}(p)$ as the $ij$-th coefficient of matrix $\Phi(t_p,t_{p+1})$. So, the $\phi_{ij}(p)$ are unique and equation~\eqref{eq:discrete_process_phi} is satisfied. Moreover, for given weights $a_{ij}(t)$, the matrix $\Phi(t,T)$ is independent of the state $x(t)$ and so are the weights $\phi_{ij}(p)$. So if we assume artificial states $x_j(t_p) = 1$ for $j\in S$ and $x_k(t_p) = 0$ for $k\notin S$, we obtain \eqref{eq:phi_artificial_states} from equation~\eqref{eq:discrete_process_phi}. And since system~\eqref{eq:sys-integral} preserves the nonnegativity of the states,
% property that if all states are non-negative, they remain non-negative~\cite{??}, 
it follows from equation~\eqref{eq:phi_artificial_states} applied to $S = \{j\}$ that $\phi_{ij}(p)\geq 0$ for every $i,j,p$.

Finally, we can use the Peano-Baker formula~\cite{brockett1970finite}
%formula~\cite{??} ADD REFERENCE LATER
to show that $\sum_{j \in \NN} \phi_{ij}(p)= 1$ : the formula gives $\Phi(t,T)$ as the limit of a recursive series
$$
\Phi(t,T) = \lim_{n\rightarrow \infty} M_n(T)
$$
\vspace{-0.2cm}
with
\vspace{-0.1cm}
$$
M_0(\tau) = I \text{ and } M_{n+1}(\tau) = I-\int_t^\tau L(s) M_n(s)ds,
$$
where $I$ is the identity matrix and $L(s)$ the Laplacian matrix of $A(s) = (a_{ij}(s))$, \ie with diagonal elements equal to $\sum_{j\in \NN} a_{ij}(s)$ and off-diagonal elements equal to $-a_{ij}(s)$.
Since $L\cdot \un = 0$ with $\un$ the vector of all ones, we have from the recursive equation that $M_n\cdot 1 = 1$ and by continuity, $\Phi(t,T) \cdot 1 = 1$, thus
$\sum_{j \in \NN} \phi_{ij}(p)= 1$. 
\end{proof}

%It can be shown that for given $a_{ij}(t)$, there exist unique weights
%$\phi_{ij}(p)$ exist and are unique, 
%which are directly determined by the values of the coefficients $a_{ij}(t)$ over $[t_p,t_{p+1}]$,
%are nonnegative, and satisfy $\ssum_{j \in \NN} \phi_{ij}(p)= 1$. (It can be shown that such $\phi_{ij}(p)$ exists and are unique). 

To obtain more insight on the discrete-time weights $\phi_{ij}$, we give the next proposition which bounds the discrete-time weights $\phi_{ij}$ using the continuous-time weights $a_{ij}$. For concision, we will omit the explicit reference to time in $a_{ij}$ when the context prevents any ambiguity. 
%the sequel we will
\vspace{.2cm}

\begin{proposition}\label{prop:bound-discrete-time-weights-vs-continuous-time-weights}
Under the uniform bound Assumption~\ref{as:upper-bound-on-weights}, we have for all proper subset of agents $S$ and all $p\ge 0$,
$$
G \cdot \ssum_{\samf{\iSjnotS}}
%\begin{minipage}{0.8cm}{\begin{center}\scriptsize{$i\in S$ \\ $j\notin S$}\end{center}}\end{minipage}} 
\int_{t_p}^{t_{p+1}} a_{ij}(t)dt \le \ssum_{\samf{\iSjnotS}} \phi_{ij}(p) \le n \cdot \ssum_{\samf{\iSjnotS}} \int_{t_p}^{t_{p+1}} a_{ij}(t)dt,
$$
with $G = \exp(-2nM)/n$.
\end{proposition}

\vspace{.2cm}
\begin{proof}
%[of Proposition~\ref{prop:bound-discrete-time-weights-vs-continuous-time-weights}]
Let $p\in \N$ and $S$ a proper subset of $\NN$. We assume that 
%\text{for all }
\begin{equation}\label{state}
\forall i \in S,  x_i(t_p) = 0 \text{ and } \forall j \in S,  x_j(t_p) = 1,
\end{equation}
as suggested in Remark \ref{rem:artificial}.
% Lemma~\ref{l:discrete-time-system}.

We first show the left inequality. We show that starting from state~(\ref{state}) at time $t_p$ no agent $j \notin S$ can be arbitrarily close to $0$ at time $t_{p+1}$. We have for all $\tau \in [t_p, t_{p+1}]$,
\begin{eqnarray*}
x_j(\tau) &=& x_j(t_{p}) + \int_{t_p}^{\tau} \ssum_{k\in \NN} a_{jk}(t) \cdot (x_k(t)-x_j(t))dt \\
&\ge& x_j(t_{p}) - \int_{t_p}^{\tau} \ssum_{k\in \NN}a_{jk}(t) \cdot x_j(t) dt,
\end{eqnarray*}
where we used $x_k(t)\ge  0, k\in \NN$. We use Gronwall's inequality~\samf{\cite{ames1997inequalities}}
%\julf{\comjh{Ideally say where in the book}}
and Assumption~\ref{as:upper-bound-on-weights} (upper bound on interactions on each $[t_p,t_{p+1}]$) to obtain
\begin{equation}\label{eq:lower-bound-on-xj}
j\not\in S\Rightarrow x_j(\tau) \ge e^{-nM}, \forall \tau \in [t_p,t_{p+1}].
\end{equation}
We will use the bound~(\ref{eq:lower-bound-on-xj}) to establish that, due to attraction from agents not in $S$, the states $x_i(t_{p+1})$ of the agents $ i\in S $ at time $t_{p+1}$ are all at least a certain positive distance from $0$.
%We rely on bound~(\ref{eq:lower-bound-on-xj}) to prove that, due to attraction from agents not in $S$, all states $x_i(t_{p+1})$, for $ i\in S $, cannot be arbitrarily close to $0$ at time $t_{p+1}$.
\vspace{.2cm}

Let now $h \in S$ be such that
$$\ssum_{j\notin S} \int_{t_p}^{t_{p+1}} a_{hj}(t)dt = \max_{i\in S} \ssum_{j\notin S} \int_{t_p}^{t_{p+1}} a_{ij}(t)dt,$$
\ie agent $h$ is the element in $S$ receiving the highest influence from the rest of the group. There holds
\begin{equation}\label{eq:lower_bound_influence_h}
\ssum_{j\notin S} \int_{t_p}^{t_{p+1}} a_{hj}(t)dt \geq \frac{1}{n} \sum_{i\in S}\ssum_{j\notin S} \int_{t_p}^{t_{p+1}} a_{ij}(t)dt.
\end{equation}
Using the non-negativity $x_i \ge 0$ for $i \in S$ and the lower bound (\ref{eq:lower-bound-on-xj}) on $x_j$ for $j \notin S$, we have for all $\tau \in [t_p, t_{p+1}]$,
\begin{small}
\begin{equation*}
 \begin{array}{l}
  x_h(\tau) =\\
  x_h(t_{p}) +  \int_{t_p}^{\tau} \ssum_{j\notin S} a_{hj} (x_j-x_h)dt + \int_{t_p}^{\tau} \ssum_{k \in \NN} a_{hk} (x_k-x_h) dt\\
 \geq  x_h(t_{p}) + \int_{t_p}^{\tau} \ssum_{j\notin S} a_{hj} x_j dt
- \int_{t_p}^{\tau} \ssum_{k \in \NN} a_{hk} x_h dt.\\
 \geq   e^{-nM} \int_{t_p}^{\tau} \ssum_{j\notin S} a_{hj} dt
- \int_{t_p}^{\tau} \ssum_{k \in \NN} a_{hk} x_h dt,
 \end{array}
\end{equation*}\end{small}
%\begin{eqnarray*}
%%$$
%x_h(\tau) &=& x_h(t_{p}) +  \int_{t_p}^{\tau} \ssum_{j\notin S} a_{hj} x_j
%- \int_{t_p}^{\tau} \ssum_{k \in \NN} a_{hk} x_h.\\
% &=& x_h(t_{p}) +  e^{-nM} \int_{t_p}^{\tau} \ssum_{j\notin S} a_{hj}
%- \int_{t_p}^{\tau} \ssum_{k \in \NN} a_{hk} x_h.
%%$$
%\end{eqnarray*}
where we have also used $x_h(t_p) = 0$. It follows then from Gronwall's inequality that %%%% Say which one in the paper version that 
%We again use Gronwall's inequality and that $x_h(t_p) = 0$ to obtain
\begin{equation}
x_h(t_{p+1}) \ge 
%e^{-\int_{t_p}^{t_{p+1}} \sum_{k\in \NN} a_{hk}} x_h(t_p)
e^{-nM} \int_{t_p}^{t_{p+1}}  e^{-\int_{\tau}^{t_{p+1}} \sum_{k\in \NN} a_{hk}ds }  \ssum_{j\notin S} a_{hj}d\tau.
\end{equation}
The expression inside the exponential can be bounded
%We now bound the term in the exponential 
using Assumption~\ref{as:upper-bound-on-weights} (upper bound) together with \eqref{eq:lower_bound_influence_h}. We have then
% to write
\begin{equation}
x_h(t_{p+1}) \ge 
%e^{-\int_{t_p}^{t_{p+1}} \sum_{k\in \NN} a_{hk}} x_h(t_p)
\frac{1}{n}e^{-2nM} \ssum_{i\in S,j\notin S} \int_{t_p}^{t_{p+1}}  a_{ij}dt.
\end{equation}
Moreover, $\phi_{ij}(p)\ge 0$ and equation~\eqref{eq:phi_artificial_states} yield
$$
\ssum_{i\in S}\ssum_{j\notin S}\phi_{ij}(p) \ge \ssum_{j\notin S}\phi_{hj}(p) = x_h(t_{p+1}).
$$
We conclude the first part of the proof combining the two previous equations.
\vspace{.4cm}

We now turn to the second inequality. For $t\in [t_p, t_{p+1}]$, let $\xM_S(t) = \max_{i\in S} x_i(t)$ be the largest value in $x$ at time $t$, and $m(t)$ the index of (one of) the agents holding that largest value at time $t$. It was shown in~\cite[Proposition 2]{Hendrickx2011} that

$$
\xM_S(t) = \xM_S(t_p) + \int_{t_p}^{t} \ssum_{k \in \NN} a_{m(\tau)k} (x_k-\xM_S)d\tau.
$$
%where $m(\tau) \in S$ is chosen such that $x_{m(\tau)}(\tau) = \xM_S(\tau)$. The last equality has been shown in~\cite[Proposition 2]{hendrickx2011new}. 
Notice that the choice of state~(\ref{state}) implies $\xM_S(t_p)=0$. Since $x_j \le 1$ for $j \notin S$ and $x_i \le \xM_S \le 1$ for $i \in S$, we have
\begin{eqnarray*}
\xM_S(t) &\le& \int_{t_p}^{t} \ssum_{j \notin S} a_{m(\tau)j} (1-\xM_S)d\tau \\
&\le& \int_{t_p}^{t}  \ssum_{i\in S, j \notin S} a_{ij} (1-\xM_S)d\tau.
\end{eqnarray*}
%which rewrites to
%$$
%1- \xM_S(t) &\ge& 1-\xM_S(t_p) -  \int_{t_p}^{t}  \ssum_{i\in S, j \notin S} a_{ij} (1-\xM_S).
%$$
Gronwall's inequality yields then %and Assumption~\ref{as:upper-bound-on-weights} (upper bound) 
%yields
\begin{align}\label{eq:intermediary-bound}
\xM_S(t_{p+1}) &\le 1- \exp\left(-\int_{t_p}^{t_{p+1}}  \sum_{i\in S, j \notin S} a_{ij} dt\right) \nonumber \\
&\le \int_{t_p}^{t_{p+1}}  \ssum_{i\in S, j \notin S} a_{ij} dt,
\end{align}
from which we conclude
%We conclude with
\begin{center}
$
\ssum_{i\in S,j\notin S}\phi_{ij}(p) = \ssum_{i\in S}x_i(t_{p+1}) \le n \xM_S(t_{p+1}).
$
\end{center}
\end{proof}

The previous proposition serves to transpose the cut-balance assumption provided in Theorem~\ref{th:consensus-under-persistent-connectivity-and-upper-bound-on-cummulative-weights-and-reciprocity} to the discrete-time weights $\phi_{ij}(p)$. In particular, we can now show that the weights $\phi_{ij}(p)$ satisfy the condition of Theorem \ref{thm:DT_result}.
% are satisfied, more precisely, we have Lemma~\ref{l:a-and-b} regarding weights $\phi_{ij}(p)$.

\vspace{.2cm}
\begin{lemma}\label{l:a-and-b}
 \sam{Under the non-instantaneous reciprocity Assumption~\ref{as:reciprocity} and the uniform bound Assumption~\ref{as:upper-bound-on-weights}}, the following properties hold :
 \begin{itemize}
 \item[a)] There exists a uniform lower bound $\beta >0$ on diagonal elements: $\phi_{ii}(p) \ge \beta$, for all $p$ and $i$.
  \item[b)] The weights $\phi_{ij}(p)$ satisfy the cut balance assumption (\ref{eq:cut_balance_assumption_DT}) for some $K'$ determined by the constants $K$ and $M$ of Assumptions \ref{as:reciprocity} and \ref{as:upper-bound-on-weights}.
  %not neccessarily the cut-balanced assumption : there exists $Q$ such that for all proper subset $S$ of agents, and all $p\ge 0$, $$\sum_{i \notin S,j \in S} \phi_{ij}(p) \le Q \cdot \sum_{i \notin S,j \in S} \phi_{ji}(p),$$
\end{itemize}
\end{lemma}
\vspace{.2 cm}
Note that \jmh{(b)} would in general not be true for certain stronger forms or reciprocity. In particular, $\int_{t_p}^{t_{p+1}} a_{ij}(t) dt \leq K \int_{t_p}^{t_{p+1}} a_{ji}(t) dt$ does not imply the existence of a $K'$ such that $\phi_{ij}(p)\leq K' \phi_{ji}(p)$.
\vspace{.2 cm}

%\end{proof}

%\noindent\textbf{Proof of a).}
\begin{proof}
The proof of (a) is as follows. For arbitrary $k\in \NN$ and $p$, we suppose that $x_k(t_p)=0$, and $x_i(t_p) =1$ for every $i\neq k$. A reasoning similar to that leading to \eqref{eq:lower-bound-on-xj} in the proof of  Proposition~\ref{prop:bound-discrete-time-weights-vs-continuous-time-weights} shows that $x_k(t_{p+1})\leq 1- e^{-nM}$. It follows then from Lemma \ref{l:discrete-time-system} applied to $S=\{k\}$ that $\sum_{j\in \NN,j\neq k} \phi_{kj}(p) \leq 1- e^{-nM}$, and thus that $\phi_{kk}(p)\geq e^{-nM}$, which establishes (a).

%The proof of a) is as follows. For some $k \in \NN$, denote $S = \{k\}$. We apply the first inequality in equation~(\ref{eq:intermediary-bound}) and Assumption~\ref{as:upper-bound-on-weights} (upper bound) to show that
%$$
%x_k(t_{p+1}) \le 1- e^{-nM}.
%$$
%Then,
%$$
%\phi_{kk}(p) = 1- \ssum_{j\neq k} \phi_{kj}(p) = 1-x_k(t_{p+1}) \ge e^{-nM}.
%$$

%\noindent\textbf{Proof of b).}
We now prove statement (b). The first inequality of Proposition~\ref{prop:bound-discrete-time-weights-vs-continuous-time-weights} applied to $S$ states that
\begin{equation}\label{eq:cbdirect}
\ssum_{i\in S}^{j\notin S} \phi_{ij}(p) \le n \cdot \ssum_{i\in S}^{j\notin S} \int_{t_p}^{t_{p+1}} a_{ij}(t)dt.
\end{equation}
On the other hand, 
%applying 
the second inequality of the same proposition applied to $\NN \setminus S$ yields
$$
G \cdot \ssum_{i\not \in S}^{j\in S} \int_{t_p}^{t_{p+1}} a_{ij}(t)dt \le \ssum_{i\not \in S}^{j\in S} \phi_{ij}(p),
$$
which can be rewritten as
\begin{equation}\label{eq:cb_indirect}
G \cdot \ssum_{i \in S}^{j\not \in S} \int_{t_p}^{t_{p+1}} a_{ji}(t)dt \le \ssum_{i\in S}^{j\notin S} \phi_{ji}(p).
\end{equation}
Statement (b) with $K' = n/G$ follows then directly from Assumption \ref{as:reciprocity} and the inequalities \eqref{eq:cbdirect} and \eqref{eq:cb_indirect}.
% is a direct corollary of Proposition~\ref{prop:bound-discrete-time-weights-vs-continuous-time-weights} applied to $S$ and $\NN\setminus S$, with $K' = n/G$.
\end{proof}
\vspace{.2cm}
\begin{proof}[of Theorem~\ref{th:consensus-under-persistent-connectivity-and-upper-bound-on-cummulative-weights-and-reciprocity}]

Since Lemma~\ref{l:a-and-b} is satisfied, Theorem~\ref{thm:DT_result} applies.
Thus, the sequence $x(t_p)$ converges to some $x^*\in \Re^n$. Denote by $G'=(\NN,E')$ the directed graph where $(j,i)\in E$ if
%whenever 
$\sum_{p=0}^\infty \varphi_{ij}(p) = +\infty$.
Theorem \ref{thm:DT_result} implies that the (strongly) connected components of $G'$ are entirely disconnected from each other (i.e. the different strongly connected components are not joined by any edge), and that $x_i^* = x_j^*$ if $i$ and $j$ belong to the same component. A local consensus takes thus place on each such component for the discrete-time system $y(p) = x(t_p)$. 
%Theorem \ref{thm:DT_result} also implies that $x_i^* = x_j^*$ if $i$ and $j$ belong to the same connected component of the graph $G'$. This graph only has strong components which are fully disconnected to each other and in which consensus takes place for the discrete-time system $y(p) = x(t_p)$. 
Now, the graph $G$ of persistent interactions defined in the statement of Theorem \ref{th:consensus-under-persistent-connectivity-and-upper-bound-on-cummulative-weights-and-reciprocity} is in general different from $G'$. However,
as a direct corollary of Proposition~\ref{prop:bound-discrete-time-weights-vs-continuous-time-weights}, we have that $G$ and $G'$ have the same connected components.

It remains to show that the continuous-time function $x(t)$ converges to the same $x^*$ as the sequence $(x(t_p))$.
%\noindent\textbf{Proof of d).}
We prove this by showing that for each set of node $S \subseteq \NN$ inducing strongly connected component in $G$ (or in $G'$), both the minimum $\xm_S(t) = \min_{i \in S} x_i(t)$ and the maximum $\xM_S(t) = \max_{i \in S} x_i(t)$
% and the diameter $\xM_S(t) - \xm_S(t)$ all 
converge to the same value.
Since $S$ is a connected component of $G$, the integral influence $\sum_{i\in S, j\notin S}\int_0^\infty a_{ij}(t) dt$ is finite. For any $\mu > 0$, there exists some $T_\mu\ge 0$ such that
$$
\sum_{i\in S, j\notin S}\int_{T_\mu}^\infty a_{ij}(t)dt < \mu.
$$
We denote again by $m(\tau)$ the index of (one of) the agents with the largest value, so that $x_{m(\tau)}(\tau) = \xM_S(\tau)$. Then, for all $v>u \ge T_\mu$, we have
%For all $v>u \ge T_\mu$, we have, using notation $m(\tau) \in S$ chosen such that $x_{m(\tau)}(\tau) = \xM_S(\tau)$ as before,
\begin{eqnarray*}
\xM_S(v) - \xM_S(u) &\le& \sum_{j\notin S}\int_{u}^v a_{m(\tau)j} (x_j(\tau)-x_{m(\tau)}(\tau)) d\tau \\
&\le& \sum_{j\notin S}\int_{u}^v a_{m(\tau)j} |x_j(\tau)-x_{m(\tau)}(\tau)| d\tau \\
&\le& \sum_{i \in S, j\notin S}\int_{u}^v a_{m(\tau)j} |x_j(\tau)-x_i(\tau)| d\tau \\
&\le& \sum_{i\in S, j\notin S}\int_{T_\mu}^\infty a_{ij} | x_j(\tau)-x_i(\tau) | d\tau \\
&\le& \mu \Diam(0), \\
\end{eqnarray*}
where $\Diam(0) = \max_{i \in \NN} x_i(0) - \min_{i \in \NN} x_i(0)$. This shows that the $\xM_S$ form a Cauchy sequence and thus that they converge. %converges in the sense of Cauchy, thus it converges. 
Since the sub-sequence $(\xM_S(t_p))$ converges to $x_i^*$ for some $i\in S$, there holds $\lim_{t\rightarrow +\infty} \xM_S(t) = x_i^*$. We can apply the same reasoning to show that $\xm_S$ also converges $\lim_{p\rightarrow +\infty} \xm_S(t_p) = x_i^*$. We conclude that for all $i \in S$, $x_i(t)$ converge to the same limit $x_i^*$.
\end{proof}

\sam{\subsection{Proof of Theorem \ref{th:pairwise_reciprocity}}
\label{sec:proof_pairwise}}

For concision, we say that an unordered pair $\{i,j\}=\{j,i\}$ is \emph{active} over an interval $I$ if $\int_{t\in I}a_{ij}(t) dt\geq \eps$ and $\int_{t\in I}a_{ji}(t) dt\geq \eps$. \sam{We let $T = \max_{i,j}T_{ij}$.} The following observation compiles some properties that follow directly from the definition of being active.
\vspace{.15cm}

\sam{
\begin{observation}\label{obs:properties_active}$ $\\
\noindent a) Consider two intervals $I,J$ with $I\subseteq J$. If $\{i,j\}$ is active over $I$, it is active over $J$.\\
%\noindent b) Consider two intervals $I,J$ with $I\subseteq J$. If $\{i,j\}$ is not active over $J$, it is not active over $I$.\\
\noindent b) Under Assumption \ref{as:pairwise_reciprocity}, if $a_{ij}(t)>0$, then $\ivj$ is active over $[t-T,t+T]$.\\
\noindent c) Under Assumption \ref{as:pairwise_reciprocity}, if $\ivj$ is not active over $[t,t']$, then $a_{ij}(s) = a_{ji}(s) = 0$ for all $s\in [t+T,t'-T]$.
\end{observation}\vspace{.15cm}
}

The next Proposition is the core of our proof, it allows building a sequence of times $t_k$ valid for Assumptions \ref{as:reciprocity} and \ref{as:upper-bound-on-weights}.

\vspace{.15cm}

\begin{proposition}\label{prop:build_tk}
Suppose that Assumption~\ref{as:pairwise_reciprocity} is satisfied, and let $M =  M_1+M_2$ where $M_1,M_2$ are any constant satisfying
\begin{equation}\label{eq:defM2-M1}
M_2>n(n-1)T + T \text{ and } M_1 \geq M_2 +T.
\end{equation}
Then, there exists a sequence $t_0,t_1,\dots$ with $t_0=0$, and $t_{k+1}-t_k \leq M$, 
%with $M=M_1+M_2$ and $M_1,M_2$ being some constant satisfying
%and 
such that the following condition $\Ak$ holds for every $k$.
% where\\
\sam{
%\hspace{0.5cm}
\begin{wholeindent}
$\Ak :$ for all $i,j \in \NN$ distinct, \\ Condition $\Auk$ or Condition $\Adk$ holds, 
\end{wholeindent}
}
with
$$
\left[
\begin{array}{l}
\Auk : \forall t\in [t_k,t_{k+1}], a_{ij}(t)=0, \\
\Adk : \ivj \text{ is \on over } [t_k,t_{k+1}].
\end{array}
\right.
$$
\end{proposition}
\vspace{.15cm}
The proof of Proposition~\ref{prop:build_tk} is based on an induction that makes use of
the intermediate Condition $\Bk$ :
\sam{
\begin{wholeindent}
$\Bk :$ for all $i,j \in \NN$ distinct, \\ Condition $\Buk$ or Condition $\Bdk$ holds,
\end{wholeindent}
}
with
$$
\left[
\begin{array}{l}
\Buk : \forall t\in [t_k,t_k+T], a_{ij}(t)=0,  \\
\Bdk : \ivj \text{is \on over } [t_k,t_k + M_1].
\end{array}
\right.
$$
%\comjh{I've removed the "and" since it is a "or"}
%\comjh{I've also changed 'on' by active', and removed 'off', because it was unclear if it meant = 0 everywhere on the interval or not active}
%\vspace{.15cm}

The next Lemma treats the initialization of the induction.
\begin{lemma}\label{lem:initiate_induction}
Suppose that $a_{ij}(t)=0$ for all $t\leq 0$, and let $t_0 = 0$. Then Condition $\Bz$ holds.
\end{lemma}

\begin{proof}
Suppose that $\Buz$ does not hold, \ie $a_{ij}(t) >0$ for some $t\in [0,T]$. Then, Observation \ref{obs:properties_active}(b) implies that $\ivj$ is \on over $[t-T,t+T]$ which by Observation \ref{obs:properties_active}(a) implies that $\ivj$ is $\on$ over $[\min(0,t-T),2T]$. Since $a_{ij}(t')=0$ for all $t'<0$, it follows then that $\ivj$ is \on over $[0,2T]$ and since $M_1\ge M_2 +T > n(n-1)+2T \ge 2T$ holds according to equation~\eqref{eq:defM2-M1}, $\ivj$ is \on over $[0,M_1]$ (again thanks to Observation \ref{obs:properties_active}(a)). Thus $\Bdz$ holds and so does $\Bz$.
\end{proof}

\begin{proposition}[Inductive case]\label{prop:inductive_case}
If there exists $t_k$ such that Condition $\Bk$ holds, then there exists $t_{k+1}\leq t_k + M_1+M_2$ for which Conditions $\Ak$ and $\Bkk$ hold.
\end{proposition}
\begin{proof}

Let us introduce the following two sets of unordered pairs of agents for every $t\in [t_k,t_k+M_1+M_2]$.
\begin{itemize}
\item $R_t \subseteq \{\ivj | i,j \in \NN, i\neq j\}$: set of pairs $\ivj$ which are \on over time interval $[t_k,t]$.
%for which there exists $[\utij,\otij]\subseteq [t_k,t]$ on which $(i,j)=(j,i)$ is \activee, with $\otij-\utij \leq T$. (i.e. set of those pairs that are already \activee at $t$).
\item $V_t \subseteq \{\ivj | i,j \in \NN, i\neq j\}$: set of pairs $\ivj$ for which $a_{ij}(t') = a_{ji}(t') =0$ for all $t'\in [t,t_k+M_1+M_2]$, i.e. set of pairs where there is no interaction between $t$ and $t_k+M_1+M_2$.
\end{itemize}
Note that for all $t_k\leq t\leq s \leq t_k+M_1+M_2$, there holds $R_t \subseteq R_s$ and $V_t \subseteq V_s$, so that these sets are non-decreasing with time. The non-decrease of $V_s$ is trivial while that of $R_s$ follows directly from Observation \ref{obs:properties_active}(a).

\samf{For given $k$ and $t_k$}, we now build a $t_{k+1}$ using Algorithm \ref{algo}, which we prove to always successfully terminate. We first prove that Claims 1 and 2 hold, and then show how this implies the statement of Proposition~\ref{prop:inductive_case}.

\begin{algorithm}
  \caption{Selection of $t_{k+1}$
  \label{algo}}
  \begin{algorithmic}
    
    \Require{$t_k$ satisfies $\Bk$}\\
\hspace{-0.35cm}Set $\bart = t_k + M_1$\\
\hspace{-0.35cm}\textbf{Switch over cases 0 to 3 :}\\
    
    \textbf{Case 0}: $\bart \geq t_k+ M_1+M_2 - T$: STOP, FAILURE\\
    
    \textbf{Case 1}: Conditions $\Ak$ and $\Bkk$ are satisfied taking $t_{k+1} = \bart$. STOP, SUCCESS.\\

\textbf{Case 2}: Condition $\Ak$ does not hold taking $t_{k+1 }= \bart$.\\
\hspace{0.5cm}\emph{\textbf {Claim 1}:} There exists $\ivj\not \in R_{\bart}$ belonging to $R_{\bart+T}$.\\
{\bf Then} set $\bart=\bart+T$ and iterate.\\
%\del{The two following claims hold.\\}
%\del{
%\hspace{0.5cm}\emph{\textbf {Claim 1}:} There exists $\ivj\not \in R_{\bart}$ such that $a_{ij}(t) >0 $ for some $t\in [\bart-T,\bar T]$.\\
%\hspace{0.5cm}\emph{\textbf {Claim 2}:} The pair $\ivj$ of claim 1 belongs to $R_{\bart+T}$}
%\comjh{I've merged claims}\\

\textbf{Case 3}: Condition $\Bkk$ does not hold taking $t_{k+1} = \bart$. \\
\hspace{0.5cm}\emph{\textbf {Claim 2}:} There exists $\ivj \notin V_{\bart}$ belonging to  $V_{\bart+T}$.\\
{\bf Then} set $\bart=\bart+T$ and iterate.
%\del{The two following claims hold.}
%\\
%\del{\hspace{0.5cm}\emph{\textbf {Claim 3}:} There exists $\ivj \notin V_{\bart}$ s.t. $a_{ij}(t) >0$ for some $t\in [\bart, \bart+T]$, and $\ivj$ is not \on over $[\bart, \bart + M_1]$.\\
%there is no $[\utij,\otij]\subseteq [\bart, \bart + M_1]$ on which $(i,j)$ is \activee (with $\otij-\utij \leq T$);\\
%\hspace{0.5cm}\emph{\textbf {Claim 4}:} The pair $\ivj$ of claim 3 belongs to $V_{\bart+T}$.}\\
%\comjh{est-ce un problème que les numéro des claims ne soient pas alignés avec ceux des cases? Si oui, on peut faire passer le case 1 en dernier}
    %\EndFunction
  \end{algorithmic}
\end{algorithm}

%\vspace{8pt}

\emph{Claim 1:}\vspace{2pt}\\
In Case 2 of Algorithm \ref{algo}, Condition $\Ak$ does not hold. \samf{Thus, there exists} $\ivj$ such that
$\Auk$ does not hold, \ie $a_{ij}(t) >0 $ for some $t\in [t_k,\bart]$, and
$\Adk$ does not hold, \ie $\ivj$ is not \on over $[t_k,\bart]$.

The fact that $\Adk$ does not hold implies by definition of $R_{\bart}$ that $\ivj \not \in R_{\bart}$. Let us now show that $t\in [\bart-T,\bart]$. The fact that $\Adk$ does not hold together with Observation \ref{obs:properties_active}(c) implies that $a_{ij}(t') = 0$ for all $t' \in [t_k+T, \bart-T]$. So either $t \in [t_k,t_k+T]$ or $t\in [\bart - T,\bart]$. We show that the first case is impossible: Since $\ivj$ is not \on over $[t_k,\bart]$, and $\bart \ge t_k + M_1$, the negation of Observation \ref{obs:properties_active}(a) implies that $\ivj$ is not \on over $[t_k,t_k+M_1]$, and thus that $\Bdk$ does not hold. However, we know by hypothesis that $\Bk$ holds. Thus, $\Buk$  holds : $t\notin [t_k,t_k+T]$, and as a consequence, $t\in [\bart - T,\bart]$. 

It follows then from Observation \ref{obs:properties_active}(b) that 
%
%From Property~\ref{pty:I_ON_implies_J_ON}, pair 
$\ivj$ is \on over $[t - T,t+T]$ and  from Observation  \ref{obs:properties_active}(a) that it is \on over $[\bart - 2T,\bart+T]$. Since $\bart \ge t_k + M_1>t_k+2T$,
%\del{, and according to equation~\eqref{eq:defM2-M1}, $M_1 > n(n-1)+2T$,} 
the pair $\ivj$ is \on over $[t_k,\bart+T]$, so that : $\ivj \in R_{\bart+T}$, which achieves proving claim 1.
%\comjh{J'ai remis la preuve dans l'ordre logique, mais ca se discute}

\vspace{8pt}
\emph{Claim 2:}\vspace{2pt}\\
Since Condition $\Bkk$ does not hold, there is a pair $\ivj$ that satisfies neither $\Bukk$ nor $\Bdkk$, that is, one for which $a_{ij}(t)>0$ for some $t\in [\bart, \bart+T]$, and for which $\ivj$ is not \on over $[\bart, \bart + M_1]$.
%there is no $[\utij,\otij]\subseteq [\bart, \bart + M_1]$ on which $(i,j)$ is \activee (with $\otij-\utij \leq T$). 
Since $\bart\leq t_k + M_1+M_2 - T$ for otherwise we would have been in case 0, the $t\in [\bart, \bart+T]$ for which 
$a_{ij}(t)>0$ lies in $[\bart,t_k+M_1+M_2]$, which implies that $\ivj \not \in V_{\bart}$ by definition of $V_{\bart}$. We now show that it belongs to $V_{\bart+T}$

By Observation \ref{obs:properties_active}(c), since $\ivj$ is not \on over $[\bart, \bart + M_1]$, there holds $a_{ij}(t') = a_{ji}(t') = 0$ for all $t' \in [\bart+T,\bart + M_1-T]$. Also, by definition, $\bart \ge t_k +M_1$ and $M_1 \ge M_2+T$, so that
$$
\bart +M_1 -T \ge t_k + M_1 + M_1 -T \ge t_k + M_1 + M_2.
$$
Thus, $a_{ij}(t') = a_{ji}(t') = 0$ for all $t' \in [\bart+T,t_k + M_1 + M_2]$ and $\ivj \in V_{\bart+T}$.\\

To complete the proof of Proposition \ref{prop:inductive_case}, we  show that Algorithm~\ref{algo} stops and that when it does, the choice $t_{k+1} = \bart$ satisfies Conditions $\Ak$ and $\Bkk$.
Indeed, at every iteration, either the algorithm stops, or case 2 or 3 applies and $\bart$ increases by $T$. 
%If case 2 or 3 applies, $\bart$ increases by $T$. Otherwise the algorithm stops. 
In case 2, it follows from Claim 1 that the size of $R_{\bart}$ increases by at least $1$, and in case 3, it follows from Claim 2 that the size of $V_{\bart}$ increases by at least $1$. Since both $R_{\bart}$ and $V_{\bart}$ are sets of unordered pairs of distinct nodes, their size cannot exceed $n(n-1)/2$. Therefore, cases 2 and 3 do not apply more than $n(n-1)/2$ times each. In particular, case 0 or 1 must apply for some $\bart \le t_k + M_1 + n(n-1)T$ (remembering that $\bart$ is initially $t_k+M_1$), at which stage the algorithm stops. Now since according to equation~\eqref{eq:defM2-M1}, $M_2> n(n-1)T +T $, case 0 or 1 apply for $\bart < t_k + M_1 + M_2 - T$, so that case 1 must apply first, and the algorithm produces thus a $t_{k+1}$ satisfying $t_{k+1} - t_k \le M_1 + M_2$ for which $\Ak$ and $\Bkk$ are satisfied.
\end{proof}
%,  before or at $t_k + M_1 + n(n-1)/2$, so that case 1 must apply first, and t
\vspace{.15cm}

The proof of Proposition \ref{prop:build_tk} is then a direct consequence of Lemma \ref{lem:initiate_induction} and Proposition \ref{prop:inductive_case}.
\vspace{.15cm}

\begin{proof}[of Theorem \ref{th:pairwise_reciprocity}]
We show that the sequence $t_k$ built in Proposition \ref{prop:build_tk} is valid for Assumptions \ref{as:reciprocity} and \ref{as:upper-bound-on-weights}.
Observe first that since the $a_{ij}(t)$ are assumed to be uniformly bounded \sam{above} and since $t_{k+1}-t_{k}\leq M$, there clearly holds $\int_{t_k}^{t_{k+1}}a_{ij}(t)dt< M'$ for some $M'$ and all $i,j$ and $t_k$, so that Assumption \ref{as:upper-bound-on-weights} holds. Moreover, it follows from Proposition \ref{prop:build_tk} that either $a_{ij}(t) = a_{ji}(t) = 0$ for all $t\in [t_{k},t_{k+1}]$, or $\int_{t_k}^{t_{k+1}}a_{ij}(t)dt\geq \eps  $ and $\int_{t_k}^{t_{k+1}}a_{ji}(t)dt\geq \eps $. Since the latter integrals are also bounded by $M'$, there holds
$$
\int_{t_k}^{t_{k+1}}a_{ij}(t)dt \leq \frac{M'}{\eps}\int_{t_k}^{t_{k+1}}a_{ji}(t)dt,
$$
which implies that Assumption \ref{as:reciprocity} also holds.
\end{proof}

\section{Conclusion}
\label{sec:ccl}
In this paper, we have developed a new reciprocity-based convergence result for continuous-time consensus systems. This result is based on a new assumption which allows for non-instantaneous reciprocity: Unlike previous studies, we only assume that reciprocity takes place on average over contiguous time intervals. This assumption is appropriate \jmh{for} various classes of systems (including classes of broadcasting, gossiping, and self-triggered system where communication is not necessarily synchronous). \sam{We have shown that integral reciprocity (Assumption \ref{as:reciprocity}) alone is not a sufficient condition for convergence. In particular, it does not forbid certain oscillatory behaviors. We have therefore proposed a companion assumption (Assumption \ref{as:upper-bound-on-weights}) stating that quantity of interactions taking place in the intervals over which reciprocity occurs should be uniformly bounded. Under these two assumptions, we have proven that the trajectory of the system always converges, though not necessarily to consensus.} Moreover, consensus takes place among agents in clusters of the \jmh{graph of persistent interactions.} We have also particularized our result to a class of systems satisfying a local pairwise form of reciprocity.
%\del{persistent graph.}

%\comjh{+ add a small comment about why it is useful?} \comsam{done in the beginning of the above paragraph (''This assumption is...``), do you want to detail ?}

Apart from the integral reciprocity and uniform bound, our result does not make any assumption on the interactions between agents, and allows in particular for arbitrary long periods during which the system is idle. As a consequence, it is in general impossible to give absolute bounds on the speed of convergence under the assumptions that we have made. However, future works could relate the speed of convergence to the amount of interactions having taken place in the system, as in~\cite{SamAntoine_Persistent_SICON2013}.

%We believe that this new results can serve in many instances of systems where communications are asynchronous.

\section*{Acknowledgement}

The authors wish to thank Behrouz Touri for his help about Theorem \ref{thm:DT_result}.
\vspace{-0.2cm}

\bibliographystyle{IEEEtran}
%\bibliography{bibtex}

%\bibliographystyle{plain}
\bibliography{references.bib}

\end{document}